\documentclass{article}
\usepackage[]{geometry}
\usepackage{natbib}
\usepackage{algorithmic}
\usepackage{algorithm}

\usepackage{microtype}
\usepackage{graphicx}
\usepackage{subfigure}
\usepackage{booktabs} %

\usepackage{hyperref}

\usepackage{amsmath}
\usepackage{amssymb}
\usepackage{mathtools}
\usepackage{amsthm}

\usepackage[capitalize,noabbrev]{cleveref}

\usepackage[textsize=tiny]{todonotes}

\usepackage{graphicx}
\newtheorem{definition}{Definition}
\newtheorem{theorem}{Theorem}
\newtheorem{example}{Example}
\newtheorem{lemma}{Lemma}
\usepackage{makecell}
\usepackage{wrapfig}
\usepackage{tabularx}
\usepackage{booktabs}

\DeclareMathOperator*{\argmax}{arg\,max}

\usepackage{color}

\definecolor{english}{rgb}{0.0, 0.0, 0.0}
\definecolor{mgadd}{rgb}{0.0, 0.0, 0.0}

\newcommand{\MDP}{\mathcal M}

\title{Oracles \& Followers: Stackelberg Equilibria in Deep Multi-Agent Reinforcement Learning}
\usepackage{authblk}
\author[1,2]{Matthias Gerstgrasser \thanks{matthias@seas.harvard.edu}}
\author[1,3]{David C. Parkes \thanks{parkes@seas.harvard.edu}}
\affil[1]{John A. Paulson School of Engineering and Applied Science, Harvard University, Cambridge, MA, USA}
\affil[2]{Computer Science Department, Stanford University, Stanford, CA, USA}
\affil[3]{Deepmind, London, UK}

\begin{document}

\maketitle

\begin{abstract}
Stackelberg  equilibria  arise naturally in a range of popular learning problems, such as in security games or indirect mechanism design, and have received increasing attention in the reinforcement learning literature. We present a general framework for implementing Stackelberg equilibria search as a multi-agent RL problem, allowing a wide range of algorithmic design choices. We discuss how previous approaches can be seen as specific instantiations of this framework. As a key insight, we note that the design space allows for approaches not previously seen in the literature, for instance by leveraging multitask and meta-RL techniques for follower convergence. We propose one such approach using contextual policies, and evaluate it experimentally on both standard and novel benchmark domains, showing greatly improved sample efficiency compared to previous approaches. Finally, we explore the effect of adopting algorithm designs outside the borders of our framework.
\end{abstract}

\section{Introduction}

Stackelberg equilibria are an important concept in economics, and in recent years have received increasing attention in computer science and specifically in the multiagent learning community. 
They model an asymmetric setting: a leader  who commits to a strategy, and one or more followers who respond. The leader aims to maximize their reward, knowing that followers in turn will best-respond to the leader's choice of strategy. These equilibria appear in a wide range of settings. In {\em security games}, a defender wishes to choose an optimal strategy considering attackers will adapt to it \citep{an2017stackelberg,sinha2018stackelberg}. In {\em mechanism design}, a designer wants to allocate resources in an efficient manner, knowing that participants may strategize \citep{nisan1999algorithmic,swamy2007effectiveness,brero2021learning}. More broadly, many multi-agent system design problems can be viewed as Stackelberg equilibrium problems: we as the designer take on the role of the Stackelberg leader, wishing to design a system that is robust to agent behavior.

In this paper, we are particularly interested in Stackelberg equilibria in sequential decision making settings, i.e., {\em stochastic Markov games}, and using multi-agent reinforcement learning  to learn these equilibria. We make two key contributions:
\begin{enumerate}
    \item We introduce a new theoretical framework for framing Stackelberg equilibria as a multi-agent reinforcement learning problem.
    \begin{enumerate}
        \item We give a theorem (Theorem~\ref{thm:main}) that unifies
        {and generalizes} several prior algorithmic design choices in the literature. 
        \item Complementing this, we show that algorithm designs outside our framework can provably fail {(Theorem~\ref{thm:divergence})}.
        In particular, we give to our knowledge the first demonstration
        that reinforcement learning (RL) against best-responding opponents can provably diverge.%
    \end{enumerate}
    \item  Inspired by Theorem~\ref{thm:main}, we introduce a novel approach to accelerating follower best-response convergence, borrowing ideas from multitask and meta-RL, including an experimental evaluation on existing and new benchmark domains.
\end{enumerate}

Our theoretical framework (Section~\ref{sec:framework}) allows a black-box reduction from learning Stackelberg equilibria into separate leader and follower learning problems. This theory encompasses and generalizes several prior approaches from the literature, in particular \citet{brero2021learning}, and gives a large design space beyond what has been explored previously.
{In particular, our broader framework generalizes beyond learning Stackelberg equilibria along with follower learning dynamics to general, query-based follower algorithms.}
The theory  is complemented by necessary conditions, with a demonstration
of impossibility  for successful Stackelberg learning through RL {by the leader when used together with  followers that immediately best respond to the leader policy (Theorem~\ref{thm:divergence})}. 
Our practical meta-RL approach uses contextual policies, a common tool in multitask and meta-RL \citep{wang2016learning}, to the follower learning problem. This allows followers to generalize and quickly adapt to leader policies,
as demonstrated on existing benchmark domains, with  greatly improved speed of convergence compared to previous
approaches. Beyond this, 
we use this meta-RL approach to
scale up Stackelberg learning 
 beyond what has previously been shown,
to a state-of-the-art RL benchmark domain built on Atari 2600.

In the remainder of the paper, we will introduce Stackelberg equilibria and Markov games in Section~\ref{sec:prelim}. In Section~\ref{sec:framework}, we motivate and define our framework for learning Stackelberg equilibria using multi-agent RL, and discuss its scope and limitations. {We show sufficient conditions in Lemma~\ref{thm:lemma} and Theorem~\ref{thm:main}. In Theorem~\ref{thm:divergence} we show that RL against best-responding opponents can fail, which implies that the construction in Theorem~\ref{thm:main} is also necessary {for RL leaders (Appendix~\ref{sec:ap:limits:pomdp} discusses a case of non-RL leaders that does not require Theorem~\ref{thm:main}).} In Appendix~\ref{sec:ap:memory} and Appendix~\ref{sec:ap:limitations} we discuss further ablations and edge cases of our framework.} We define our novel Meta-RL approach in Section~\ref{sec:contextual}, and empirically evaluate it in Section~\ref{sec:experiments} on existing and new benchmark domains.

\subsection{Prior Work}

\paragraph{Learning Stackelberg Equilibria.}
Most prior work on Stackelberg equilibria focuses on single-shot settings such as normal-form games, a significantly simpler setting than Markov games. {Some of this work studies solving an optimization problem to 
find a Stackelberg equilibria, given an explicit description of the problem~\cite{paruchuri2008playing, xu2014computing,blum2014learning,li2022solving}.}
Among the first work to learn a Stackelberg equilibria was~\citet{letchford2009learning}, who focus on
{single-shot}
Bayesian games. 
\citet{peng2019learning} also give results for  sample access to {the payoffs of} 
matrix games.
\citet{wang2022coordinating} give an approach that differentiates 
through optimality (KKT) conditions, again for normal-form games.
\citet{bai2021sample} give lower and upper bounds on learning Stackelberg equilibria in general-sum games, including so-called ``bandit RL'' games that have one step for the leader
and sequential decision-making for the followers.
Few works consider Markov games: \citet{zhong2021can} give algorithms that find Stackelberg equilibria in Markov games, but assume myopic followers, a significant limitation compared to the general case.  \citet{brero2021learning, brero2021reinforcement, brero2022learning} use an inner-outer loop approach, which they call the {\em Stackelberg POMDP}, {which can be seen as a special case of our framework.}
\paragraph{Mechanism Design.}
One of the first works specifically discussing Stackelberg equilibria in a learning context is \citet{swamy2007effectiveness}, who design interventions in traffic patterns. More recently, several strands of work have focused on using multi-agent RL for economic design, often framing this is a bi-level or inner-outer-loop optimization problem. %
\citet{zheng2022ai} use a bi-level RL approach to design optimal tax policies, and
\citet{yang2022adaptive} use  meta-gradients in a specific incentive design setting. %
\citet{shu2018m} and \citet{shi2019learning} learn leader policies in a
type of incentive-shaping setting, using
a form of modeling other agents 
coupled with rule-based followers. \citet{balaguer2022good} use an inner-loop outer-loop, gradient descent approach for mechanism design on iterated matrix games (which we also use as an experimental testbed). They mainly focus on the case where both the environment transition as well as the follower learning behavior is differentiable, and otherwise fall back to evolutionary strategies for the leader.

\paragraph{Opponent Shaping.}
More broadly related is also a line of work on \textit{opponent shaping}, such as M-FOS~\citep{lu2022model} or LOLA~\citep{foerster2017learning}. While these works also learn while taking into account other agents' responses, they intentionally only let the other agents learn ``a little bit'' between the opponent-shaping agent's learning steps, as opposed to letting them learn until they best-respond. This reflects a difference in goals: Opponent-shaping aims to predict and exploit other agents' learning behaviors, whereas Stackelberg work aims to learn strategies that are robust even when other agents are able to learn and adapt as much as they like.

\section{Preliminaries}
\label{sec:prelim}

\paragraph{Markov games.} We consider partially observable stochastic Markov games, which are 
a multi-agent generalization of a partially observable Markov Decision Process (POMDP).
\begin{definition}[Markov Game] \label{def:markov}
A {\em Markov Game}, $\mathcal M$,
with $n$ agents is a tuple $(S,A,T,r,\Omega,O,\gamma)$, consisting of a state space $S$, an action space $A = (A_1,...,A_n)$, a (stochastic) transition function $T: S \times A \rightarrow S$, a (stochastic) reward function $r: S \times A \rightarrow \mathbb R^n$, an observation space $\Omega = (\Omega_1,...,\Omega_n)$, a (stochastic) observation function $O: S \times A \rightarrow \Omega$, and a discount factor $\gamma$.
\end{definition}

At each step $t$ of the game, every agent $i$ chooses an action $a_{i,t}$ from their action space $A_i$, the game state evolves according to the joint action $(a_{1,t},\ldots,a_{n,t})$ and the transition function $T$, and agents receive observations and reward according to $O$ and $R$. An agent's behavior in the game is characterized by its policy $\pi_i: o_i \mapsto a_i$, which maps observations to actions.\footnote{To keep notation concise we discuss here the memory-less case, but all our results generalize to a stateful leader policy in a straightforward manner, as we discuss in Appendix~\ref{sec:ap:memory}.}
Each agent in a Markov Game individually seeks to maximize its own (discounted) total return $\sum_t \gamma^t r_i(s_t, a_{i,t}, a_{-i,t})$. This gives rise to the usual definitions of Nash equilibria (NE), correlated equilibria (CE), and coarse correlated equilibria (CCE), which we do not repeat in full here, as well as their Bayesian counterparts. %
Note that strategies in Markov games and in each of these equilibrium definitions are \textit{policies}, not actions: A pair of policies $\pi_1, \pi_2$ in a two-player Markov game is a Nash equilibrium if neither agent can increase their expected total reward by unilaterally deviating to a different policy $\pi_{i}^{'}$.

\paragraph{Stackelberg Equilibria.}
Unlike  the above equilibrium concepts, a Stackelberg equilibrium is not symmetric: There is a special player, the {\em leader}, who commits to their strategy first; the other player (the {\em follower}) then chooses their best strategy given the leader's choice of strategy. This makes the leader potentially more powerful. 

\begin{wraptable}{r}{0.2\textwidth}
\caption{``Battle of \\the Sexes'' game.}
\label{matrix-bots}
\begin{center}
\begin{tabular}{|l|l|}
\hline
2,1 & 0,0 \\
\hline
0,0 & 1,2 \\
\hline
\end{tabular}
\end{center}
\end{wraptable}
\begin{example}
In a game often called the ``battle of the sexes," you and I wish to have dinner together, but you prefer restaurant A (deriving happiness $2$, but I only get happiness $1$), and I prefer restaurant B (I get happiness $2$, you get happiness $1$)---but we would both rather eat together at our less-preferred venue, than to eat separately (we both get happiness $0$). Table~\ref{matrix-bots} shows the payoff matrix of this game. There are two pure Nash equilibria in this game: We both go to restaurant A, or we both go to restaurant B. But there is only a single Stackelberg equilibrium (per leader): If you commit to going to restaurant A, then my only best response is to also go to restaurant A. In doing so I receive happiness $1$, whereas my only alternative would be to eat alone at restaurant B for happiness $0$. Notice that this hinges on the leader strictly committing to their choice of restaurant.
\end{example}

This Stackelberg concept  also extends to Markov games: Here  a leader agent $L$ decides on their policy (i.e. strategy), and the remaining (follower) agents---knowing the leader's choice of policy---best-respond. The leader seeks to maximize their own reward, considering that followers will  best-respond. For instance, in an Iterated Prisoners' Dilemma~\citep{robinson2005topology}, a leader might commit to a Tit-for-Tat strategy, 
in turn leading the follower to cooperate. 

Typically, a Stackelberg equilibrium is formally defined using a max-min-style condition: The leader maximizes its own reward knowing that the follower best-responds, i.e., maximizes its own reward, with the leader-follower dynamic giving the order of the two nested max-operators. This in turn suggests a nested outer-inner loop (reinforcement) learning approach, where the follower trains until convergence every time the leader updates its policy. As a key innovation in this work, we instead use a statement of the follower best-response through an oracle abstraction. An oracle definition has been used before in order to extend Stackelberg equilibria to multiple followers \cite{nakamura2015one, zhang2016multi, liu1998stackelberg, solis2016modeling, sinha2014finding, wang2022coordinating, brero2021learning}. 
In contrast, we use the oracle abstraction to greatly simplify the statement and proof of our main theorem, while simultaneously generalizing it beyond prior approaches. In turn, this allows us to develop a novel Meta-RL approach in the second part of the paper.

With multiple followers, any choice of leader strategy, $\pi_L$, induces a Markov game, $\mathcal F_{\pi_L}$, between the followers, which could feature multiple equilibria as well as
equilibria of different types, such as Nash, correlated, and coarse correlated equilibria, each giving rise to a corresponding Stackelberg-Nash, Stackelberg-CE, and Stackelberg-CCE concept. 
This motivates an oracle abstraction. {For any choice of leader strategy, we denote
as $\mathcal E(\mathcal F_{\pi_L})$ a follower equilibrium (or a probability distribution over equilibria), where we refer to $\mathcal E$ as an  {\em oracle}.
This oracle will later be realized  
 as an algorithm.}
See also~\citet{wang2022coordinating}.
\begin{definition}[Stackelberg equilibrium] \label{def:stackelberg}
Given a Markov Game $\mathcal M$ and a follower best-response oracle $\mathcal E$, a leader strategy $\pi_L$ together with a tuple of follower strategies $\mathbf \pi_F$ is a {\em Stackelberg equilibrium}, if {and only if}
$\pi_L$ maximizes the leader's expected reward under the condition that follower strategies are drawn from $\mathcal E(\mathcal F_{\pi_L})$:
\begin{align*}
    \pi_L \in \argmax_{\pi_L} \mathop{\mathbb{E}}_{\mathbf \pi_F \sim \mathcal E(\mathcal F_{\pi_L})} \Big[ \sum_t \mathop{\mathbb{E}}[ r_L\big( s_t, a_{L,t}, a_{F,t} \big) ] \Big],
\end{align*}
where the second expectation is drawing actions and state transitions from their respective policies $\pi_L, \pi_F$ and transition function $T$, and the reward function is $r$, all as in Definition~\ref{def:markov}.
If the follower oracle $\mathcal E$ gives a Nash equilibrium in the induced game $F_{\pi_L}$, we call this a {\em Stackelberg-Nash equilibrium}, and similarly for CE and CCE.
\end{definition}

{In the remainder of the paper, when we say ``oracle" we mean an algorithm
that  computes or learns the follower best-response equilibrium, $\mathcal E(\mathcal F_{\pi_L})$, given the leader strategy  $\pi_L$.}
In  full generality, an oracle algorithm could take many forms, including RL, optimization, or any other algorithm that takes as input $\pi_L$ and outputs $\mathcal E(\mathcal F_{\pi_L})$. {These algorithms may vary 
in regard to how they access information about $\pi_L$.
Our main positive result (Theorem~\ref{thm:main}) will rely on oracle
 algorithms that only require 
{\em query access} (also called sample access) to the leader policy; i.e.,
algorithms that only  interact with the leader policy by receiving as input the leader's actions in  each of a set of  leader observation states, these states queried by the algorithm in some order.}
\begin{definition}[Query Oracle] \label{def:queryoracle}
    An algorithm implementing an oracle $\mathcal E$ is called a {\em query oracle}
    if its interactions with the leader policy $\pi_L$ are exclusively through 
 queries, where a 
query is an input to the leader policy; i.e., 
an observation $o$ together with the response from the leader policy, 
which is the action $\pi_L(o)$ the leader would take given $o$.
\end{definition}

{A query oracle can also receive additional information, 
for example an 
interaction with the environment that is associated with the Markov Game (or 
a  description of the environment). {The definition of} a {\em query} oracle is only
in regard to how the oracle algorithm interacts with  the leader policy.}
In particular,  any oracle implemented using RL or {another typical}
learning approach
is a query oracle, as learning algorithms generally only require query access {to their environment and the actions taken by other agents in different states.}
{In particular, a suitably formulated RL leader 
together with a RL-based {followers (i.e., an ``outer loop - inner loop'' approach)} can 
be interpreted within our framework as
 an oracle querying the leader policy.} %

An algorithm that operates directly on a description of the leader policy, $\pi_L$, 
such as a parametrization $\theta$ of $\pi_L$, {is not a query oracle}.
 For instance, if $\MDP$ has a small number of discrete states and actions, $\pi_L$ could be directly parametrized,
with $\theta_{ij}$ denoting the probability of taking action $j$ in state $i$. An optimization approach could then compute a best response directly from knowledge of $\theta$, but
these parameters are not available through query access. 
Similarly, $\theta$ could be the weights of a neural network.

\section{A General Framework for Stackelberg Equilibria in Multi-Agent RL}
\label{sec:framework}

Several  approaches have  been proposed to learning Stackelberg equilibria in Markov games, or to use multi-agent RL for mechanism design in such settings. A main aim of our work is to elucidate commonalities between these approaches, and to delineate what is required to guarantee Stackelberg equilibria. For instance, most of the existing approaches use (reinforcement or no-regret) learning to implement the follower best-response, effectively arriving at an {\em ``inner-loop-outer-loop system"}: The leader performs one update to their policy, then the followers perform many updates to theirs until they converge to a best response, then this repeats. Is this the only possible approach? 
Can you mix-and-match leader and follower approaches at will? One approach for leader learning is reinforcement learning, where the gradient of the leader policy is estimated from sampled trajectories \citep{brero2021learning}---this is in contrast to global approaches such as direct differentiation of the leader policy in a world where everything is differentiable  or evolutionary strategies \citep{balaguer2022good}, which modify the leader policy as a whole based on total episode reward, without looking at what happens at each step. Some leader RL 
approaches \citep{brero2021learning,brero2022learning} incorporate the followers' learning dynamics into the leader's view of the environment. Is it necessary that the leader can see this adaption process? Or could we also have the follower best-respond to the leader immediately on the first step it takes? Some approaches for mechanism design do not explicitly mention Stackelberg equilibria \citep{zheng2022ai,balaguer2022good}, but seem very similar to approaches that do; do those approaches converge to a Stackelberg equilibrium? 
In this section we develop a common framework that answers these questions,
and provide a common language to categorize the various strands of research in this area.

A key novelty in our framework is that following Definition~\ref{def:stackelberg} we separate the problem into a leader learning problem and a follower oracle implementation. Let $\mathcal L = \mathcal L_\MDP$ be the learning problem that the leader faces. In standard multi-agent RL, $\mathcal L$ would simply be the leader's local view of the multi-agent system, considered as a single-agent learning problem (i.e., taking the other agents as being part of the environment). We will now show how to construct $\mathcal L$ so as to guarantee a Stackelberg equilibrium. As a warm-up, we will derive two basic conditions on $\mathcal L$ that guarantee that an optimal policy in $\mathcal L$ forms a Stackelberg equilibrium.
\begin{lemma} \label{thm:lemma}
Given a Markov Game $\MDP$ and a follower equilibrium oracle $\mathcal E$, let $\mathcal L_\MDP$ be the learning problem the leader faces. If:
\begin{enumerate}
    \item for each choice of leader policy $\pi_L$, $\mathcal L$ computes the follower best-response $\mathcal E(\pi_L)$, and
    \item $\mathcal L(\pi_L)$ evaluates the leader policy $\pi_L$ against the follower best-response $\mathcal E(\pi_L)$ in $\MDP$, i.e. the value of $\mathcal L(\pi_L)$ is $r_{L}(\pi_L, \mathcal E(\pi_L))$ in $\MDP$,
\end{enumerate}
then an optimal solution $\pi_L^*$ to $\mathcal L$ together with the follower best-response $\mathcal E(\pi_L^*)$ form a Stackelberg equilibrium in $\MDP$.
\end{lemma}
This follows in a straightforward way from Definition~\ref{def:stackelberg} (see Appendix~\ref{sec:ap:proof} for a formal proof).
It is also easy to see that these are essentially necessary conditions (modulo transformations and reward shaping, see Appendix~\ref{sec:ap:necessary}).
While  Lemma~\ref{thm:lemma} is simple, it already allows us to answer some of the questions we posed initially. In particular, do previous not-explicitly-Stackelberg approaches from the literature give Stackelberg equilibria? They are not guaranteed to, because many of them either do not train followers until convergence (e.g. \citep{zheng2022ai}, violating condition 1), or give reward even while followers are still learning (e.g. \citep{balaguer2022good}, violating condition 2).

However, Lemma~\ref{thm:lemma} only applies if we already have an optimal solution to $\mathcal L$, and
it does not say anything about whether a particular leader learning algorithm, such as RL, will converge on $\mathcal L$. 
In particular, we asked earlier if it is necessary for the leader to be able to see the followers' adaption process, or if we could also have followers best-respond immediately?
We  first show
that ``letting the leader see'' the followers' adaption process (formally, the operation of the oracle $\mathcal E$) makes $\mathcal L$ a POMDP, meaning that RL algorithms {for the leader} can be expected to converge under standard assumptions. 
Through the oracle abstraction, we   show this for a much broader class of follower best-response algorithms than typical learning approaches.
Second, we  show  this is also a necessary condition, i.e., 
 leader RL can fail if followers best-respond immediately
and  without a visible adaption process. {In the following theorem, we show that for \textit{any} query oracle, a suitable construction of the leader problem $\mathcal L$ will be a POMDP, i.e. exhibit Markovian state transitions.} {This POMDP  can,  in turn, be
 solved via leader RL to yield a Stackelberg equilibrium, which makes the
 theorem  applicable to settings where
it is desirable to use RL  to solve  the leader learning
 problem.}
\begin{theorem}\label{thm:main}
Given a Markov Game $\MDP$,
and a follower equilibrium oracle $\mathcal E$, if in addition to the conditions of Lemma~\ref{thm:lemma}, the follower oracle $\mathcal E$ is a \textit{query oracle} (Definition~\ref{def:queryoracle}), then the leader learning problem $\mathcal L$ can be constructed as a POMDP.
\end{theorem}

{In the following, we provide the main idea behind this positive result, which relies on the construction of a suitable, single-agent POMDP to  model  the leader's problem.
We defer the full proof of Theorem~\ref{thm:main} 
to  Appendix~\ref{sec:ap:proof}, 
where we also  give 
 concrete examples.}

{Given $\MDP$ and} {a query oracle}  {$\mathcal E$ we define {\em a new leader POMDP $\mathcal L$} as follows: The action and observation space for the leader in $\mathcal L$ are the same as those in $\MDP$. $\mathcal L$ is then constructed in two parts. First, the leader is queried} {as often as is required to compute a follower equilibrium} 
{by the oracle $\mathcal E$;
then an episode from the original Markov game $\MDP$ plays out:
\begin{itemize}
    \item \textbf{Initial Segment:} For an initial number of steps in $\mathcal L$, each step performs one query from the follower oracle $\mathcal E$: If a given query wishes to determine the leader policy's response to observation $o$, then the leader will receive $o$ as its observation in $\mathcal L$, and the leader's action will be given to $\mathcal E$ as the response to its query. The leader will receive no reward in these steps. 
    \item \textbf{Final Segment:} Once a follower equilibrium $\mathbf \pi_F$ has been determined, the remainder of $\mathcal L$ will be constructed from the original Markov game $\mathcal M$: We let followers act according to the computed follower equilibrium $\mathbf \pi_F$ and treat them (including all their internal state) as part of the environment. 
\end{itemize}}

If the follower oracle is implemented using RL, i.e., both leader and followers use RL, then the initial segment is simply one or more episodes of $\MDP$ where the followers are learning, and the final segment is one episode from $\MDP$ where the followers have converged. This formalizes{, in a general way,} the ``train followers until convergence for every leader policy update'' approach seen in  prior work~\citep{brero2022learning}. 
{The innovation is that we} 
generalize to \textit{any} {query-based} algorithm implementing the follower oracle,
not just {followers using learning algorithms},
 and this will be crucial in the next section. {(And note that we further generalize even to non-query follower algorithms in Lemma~\ref{thm:lemma}, which 
 applies to the case of non-RL leaders; see Appendix~\ref{sec:ap:limits:pomdp}).}

The POMDP property of the construction in Theorem~\ref{thm:main} gives us strong evidence that that outer-loop inner-loop RL approaches have a solid theoretical foundation.
However, such approaches also have drawbacks (such as very long and sparse-reward episodes for the leader). The remaining question is whether this visibility is necessary, or if we could skip the initial segment in the construction in Theorem~\ref{thm:main} and have followers best-respond immediately on the first step the leader takes. For the first time, we can answer this question clearly: RL against immediately-best-responding opponents can provably diverge. 
\begin{theorem} \label{thm:divergence}
    There exists a Stochastic Markov Game, $\MDP$, where neither tabular Q-learning nor policy gradient can learn the optimal policy for the leader when
    the follower agent immediately best-responds.
\end{theorem}
We prove Theorem~\ref{thm:divergence} in Appendix~\ref{sec:ap:divergence}, and  also confirm the effect experimentally in Appendix~\ref{sec:ap:limits:pomdp} in the context of the Iterated Prisoners' Dilemma. 
{This result is highly surprising, as training against a best-responding follower at first may seem like a natural approach that one might consider using. In the proof of Theorem~\ref{thm:divergence}, we see that the RL algorithm fails due to missing counterfactuals.}
{In particular, the key idea  is that the best response to a leader's policy may depend on its behavior on all possible states, including states not visited when the follower actually best-responds. For instance, in iterated prisoner's dilemma, the threat of retaliation to a defection is crucial. But if a follower best-responds by always cooperating, the defection state will never be visited, and will not lead to corresponding updates in the RL algorithm. We use this insight to then construct an instance where all other possible visits of the defection state (from a non-optimal leader policy, or random exploration) will lead to RL updates that ``point away'' from the optimal policy.}

Collectively, Lemma~\ref{thm:lemma}, Theorem~\ref{thm:main}, and Theorem~\ref{thm:divergence}  allow us to answer the last remaining question we posed initially: Can we mix-and-match arbitrary approaches  for the leader and followers? The answer is ``Yes, we can.'' The only restriction is that if we use RL for the leader (or another approach that hinges on a Markov property), then we need to use a query-based approach for the followers and the initial-final segment construction from Theorem~\ref{thm:main}.\footnote{Note that this is strictly a special case, as there are also other approaches on either side that fall outside this. For instance, one could use evolutionary strategies for the leader, coupled with a non-query-based optimization approach for the followers. See for instance Appendix~\ref{sec:ap:limits:pomdp} for a relevant experiment.}

Appendix~\ref{sec:ap:memory} further details how to extend Theorem~\ref{thm:main} to leader policies with memory, and details a crucial condition of \textit{leader invariance}. This invariance condition requires that the leader policy act the same during oracle queries as it does during real play. In the memory-less case this follows immediately in theory, but is an important implementation detail in practice. As an illustration, we show  in Appendix~\ref{sec:ap:limits:continuous} 
an example where a seemingly innocuous step counter being made part of the leader's observation leads to learning failure.

\begin{table*}[ht] \label{table:priorapproaches}
\caption{Situating different approaches within the framework of Theorem~\ref{thm:main}. Approaches marked * do not fully satisfy the conditions of Lemma~\ref{thm:lemma} and Theorem~\ref{thm:main}. Approaches marked in bold indicate approaches that apply to general Markov games, i.e., that are sequential for both leader and followers, and otherwise unrestricted.}
\small{
\label{tbl:framework_approaches}
\begin{center}
\begin{tabular}{|l|c|c|c|}
\cline{2-4}
\multicolumn{1}{c}{} & \multicolumn{3}{|c|}{{\em Leader Learning Approach}} \\
\hline
\thead{ {\em Oracle}\\ {\em Implementation}} & \thead{Optimization, \\ Search} & \thead{Direct Gradient Descent,\\ Evolutionary Strategies} & \thead{RL }  \\
\hline
\thead{N/A - Oracle\\ Assumed Given} & \makecell{\cite{letchford2009learning}\\ \cite{peng2019learning}} &  \cite{wang2022coordinating} & \cite{zhong2021can} \\
\hline
\thead{No-Regret} & & & {\cite{brero2021learning}} \\
\hline
\thead{RL} & \makecell{\cite{bai2021sample}} & \makecell{\textbf{\cite{balaguer2022good}}*\\ {\cite{yang2022adaptive}}* }& \makecell{\textbf{\cite{zheng2022ai}}* \\ \textbf{\cite{brero2022learning}} } \\
\hline
\thead{Multitask / Meta-RL} & & & \textbf{new} \\
\hline
\end{tabular}
\end{center}
}
\end{table*}

\paragraph{Relation to prior approaches.}

A key novelty in our framework is that it generalizes to arbitrary approaches for the leader and followers, allowing us to develop  the  novel Meta-RL approach in Section~\ref{sec:contextual}, 
and also supporting the unification of
several approaches from the literature. 
{As discussed above,} the main theorem of \citet{brero2021learning} can be seen as a special case of Theorem~\ref{thm:main}, with our theorem generalizing this to arbitrary leader and follower approaches. {A key difference is that their theorem focuses on follower approaches that learn in the base Markov game $\MDP$, whereas Theorem~\ref{thm:main} includes arbitrary queries of the leader policy.}

Existing approaches can   be categorized by the approaches that they take for leader and follower modeling, and are mainly
 focused on no-regret, {RL,} or policy gradient learning to implement follower oracles, coupled with either RL or direct gradient descent methods for the leader.

\citet{brero2021learning, brero2022learning} use no-regret dynamics and
 Q-learning to implement the follower oracle inside the leader's episode rollout, and standard RL techniques to solve the resulting leader POMDP. 

\citet{balaguer2022good} use gradient methods to implement the follower oracle. In the case where both followers and world dynamics are differentiable, they directly differentiate the leader policy (as opposed to estimating its gradient using sampled trajectories). For the non-differentiable case they use evolutionary computation.
Interestingly, they seem to accumulate leader reward throughout the entire learning phase of the follower. This puts their approach  outside the scope of our Theorem~\ref{thm:main}, and may give the leader the wrong optimization target, as we detail in Appendix~\ref{sec:ap:limits:reward}.
Our understanding  is that if the leader were to optimize for its final reward at the end of follower learning instead, their approach would fall within Theorem~\ref{thm:main} and yield Stackelberg equilibria. 

\citet{zheng2022ai} similarly use two-level RL design, using policy gradient to learn each of the follower oracle and  the leader policy. While they do not specifically mention Stackelberg equilibria, this would be the target equilibrium condition in the taxation-policy design setting.
Interestingly, they use a curriculum learning approach that can be seen as a rudimentary form of the contextual policy meta-learning approach that we develop in this paper. As with~\citet{balaguer2022good}, \citet{zheng2022ai} also seem to accumulate leader reward, putting them outside the assumptions of Theorem~\ref{thm:main}.
Table~\ref{table:priorapproaches} gives an overview of these prior approaches.

\section{Meta-RL for Stackelberg RL}\label{sec:contextual}

Going beyond existing approaches, Theorem~\ref{thm:main} suggests a wide design space for implementing the follower oracle. As a key contribution, we explore using multi-task and meta-RL as a means of implementing the follower oracle. This is both to illustrate the power of Theorem~\ref{thm:main} as a way to think about Stackelberg learning, as well as due to the advantages of the approach over existing ones.

We can recognize that the follower games, $\mathcal F_{\pi_{L}}$, are  in fact a family of related problems. For this reason, the follower oracle problem can be seen as a multitask or Meta-RL problem, and solved using techniques from those fields. We make use of \textit{contextual policies}~\citep{wang2016learning},  where a context $\omega$ describes the task an agent is supposed to solve. In our case, the context provides the specific MDP among a family of MDPs a follower finds itself in, and $\omega$ is a description of the leader policy. This context, $\omega$, is concatenated to the follower agent's observation $o_{i,t}$, and agent $i$ observes $(o_{i,t}, \omega)$ at timestep $t$. {Crucially, we will construct context $\omega$ through queries of the leader policy, so that we can use the POMDP construction of Theorem~\ref{thm:main}.}

We focus on settings where the leader policy's effect on the follower can be fully understood with a small number of queries, and we directly use the leader's response to a fixed set of queries as the context $\omega$. For instance, in the Iterated Prisoners' Dilemma, we ask the leader three questions: ``How do you act on the initial step of the game?'', ``How do you act if the opponent cooperated in the previous step?'' and ``How do you act if the opponent defected in the previous step?'' Clearly, if these are the only three possible states, this is sufficient to characterize the leader policy.

We further use a two-stage training approach. In Phase 1, we train a follower meta-policy against a different, randomized leader policy in each episode. By the end of this phase, the meta-policy is able to best-respond to all possible leader policies. In Phase 2, we  train a leader policy against this follower, where the leader is queried at the beginning of each episode. {Note that this two-stage approach is separate from the two-segment construction in Theorem~\ref{thm:main}. The first training phase is purely a pre-training stage for the follower meta-policy.} {In effect, it builds a suitable query-based follower oracle.}
 {The second training phase trains the leader, and uses the two-segment construction from Theorem~\ref{thm:main} in each episode experienced through leader training.} 
Concretely, in Phase 2, every episode looks as follows: First, in the ``initial segment'' (Theorem~\ref{thm:main}), the leader sees a fixed sequence of observations $o_0,\ldots,o_k$. The follower does not act in these steps at all. The leader's actions, $a_0,\ldots,a_k$, in response to the sequence of observations, $o_0,\ldots,o_k$, are memorized. Specifically in the iterated matrix game environments we use, this sequence of observations is simply all the possible states of $\MDP$, and so the leader's response fully characterizes its policy. The actions $a_0,\ldots,a_k$ then form the context $\omega = (a_0,\ldots,a_k)$ for the follower in the ``final segment.''
In this second part of each episode, leader and (meta-) follower act together in an episode of the original game, $\MDP$. At a given timestep $t$ if the current state is $s_t$, the leader's observation is $o_{L,t} = s_t$, and the (meta-) follower's observation is $o_{F,t} = (s_t, \omega)$. In words, the follower observes its usual observation, but also the context. Through this context, the follower is informed about the leader policy it is playing against, and can best-respond to it.

Because the context $\omega$ was derived through queries that are part of the leader's batch of experiences, this forms a POMDP for the leader by Theorem~\ref{thm:main}. At the same time, because the follower does not need to learn from scratch how to respond for every leader policy update, we avoid the ``inner loop'' of typical ``outer-inner loop'' approaches.
{For settings where it is not possible to explicitly define the context $\omega$ in this way,}
 the multitask and meta-RL literature provides a range of approaches that infer context, often using recurrent networks~\citep{wang2016learning, mishra2017simple, duan2016rl, rakelly2019efficient, zintgraf2019variational, humplik2019meta}.

\subsection{Experiments}
\label{sec:experiments}
We evaluate our Meta-RL approach on both a benchmark iterated matrix game domain, comparing to existing approaches, as well as on a novel Atari 2600-based domain that is significantly more challenging. In the first, our main positive finding is that our approach can match or exceed prior approaches at greatly improved sample efficiency. In the latter, we show for the first time a positive result using a principled Stackelberg approach on a state-of-the-art general RL benchmark domain. %
We will detail each of the two domains, along with the results 
obtained using the Meta-RL algorithm. Appendix~\ref{sec:ap:experiment} and~\ref{sec:ap:atari} give further details on the algorithms  and full hyperparameters.

\paragraph{Environments: Iterated Matrix Games.}
We evaluate our contextual policy approach and general framework on an ensemble of iterated symmetric matrix games, such as the Iterated Prisoners' Dilemma~\citep{robinson2005topology}. We choose these games as they present a significant step up in complexity from previous approaches that give explicit Stackelberg guarantees, in that both leader and followers face a sequential decision-making problem. 
In these, we play a matrix game for $n=10$ steps per episode, and give agents a one-step memory. This makes these environments Markov games, with five states: one for the initial steps of each episode, and four for later steps depending on the two agents' previous actions. At each step, each agent has a choice of two actions (e.g. ``cooperate'' or ``defect''), leading to the next state, e.g. ``both cooperated''.
\begin{figure*}
\begin{center}
\includegraphics[width=0.85\textwidth]{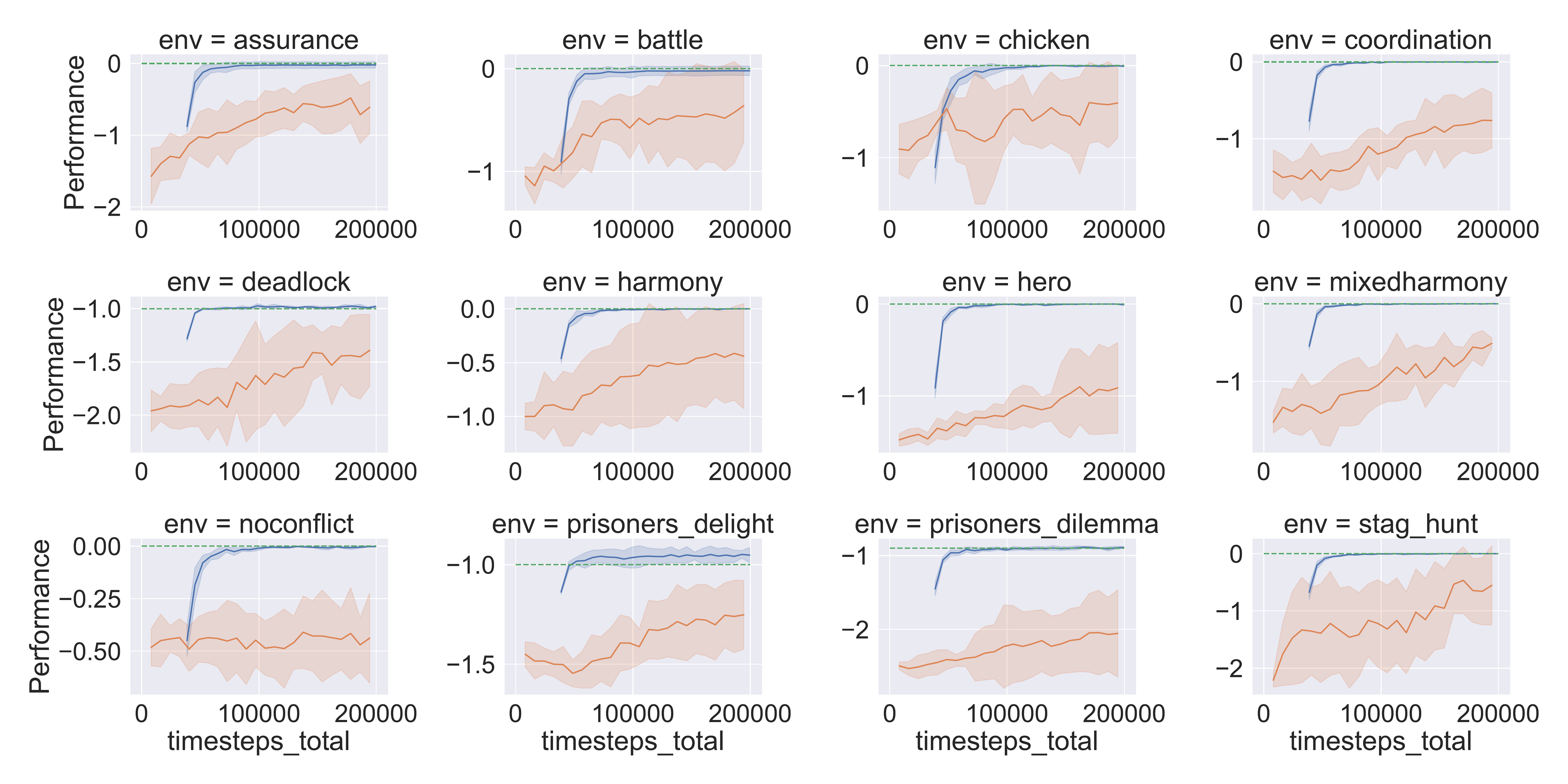}
\end{center}
\caption{Blue: Mean episode reward of our novel PPO+Meta-RL approach on 12 canonical symmetric iterated matrix games. Orange: PPO+Q-learn~\cite{brero2022learning}. Dashed green: Good Shepherd ES-MD~\cite{balaguer2022good} (final, mean episode reward  at 1.28B timesteps, estimated from Fig.~2 ibid.) 
\label{fig:allmatrices}}
\end{figure*}

Figure~\ref{fig:allmatrices} shows the performance of our Meta-RL approach using PPO for the leader. We compare against the approaches of \citet{balaguer2022good} and \citet{brero2022learning}. For our PPO+Meta-RL approach, we plot the combined environment steps used by the meta-follower training plus the leader training on the x-axis. For \citet{balaguer2022good}, we estimate performance from Figure 2 therein. Note that this is the eventual performance at the end of training (\citet{balaguer2022good} do not publish learning curves). 

In terms of mean episode reward, 
the final performance largely matches that of the ``good shepherd'' ES-MD approach, %
which is expected as both approaches achieve at or very near the theoretical optimum on all games; i.e., they both find the Stackelberg equilibrium of the game. By extension our approach also matches or outperforms all their baselines (c.f.~Figures 1 and 2 therein). The more relevant comparison between our approach and ``good shepherd'' is on speed of convergence, where
our Meta-RL approach converges in around 50k environment steps, whereas \citet{balaguer2022good} report performance at 1.28 billion environment steps in the ES-MD case. We give further details on this comparison with~\citet{balaguer2022good} in Appendix~\ref{sec:ap:performance}. 

We also see that our approach outperforms the PPO+Q-learn approach of \citet{brero2022learning}. In Appendix~\ref{sec:ap:performance} we show the PPO+Q-learn approach training for significantly longer, and see that when it does converge it does so around 500k environment steps at the earliest, whereas for most of the harder cases it still has not nearly reached optimal performance at 2M timesteps. We again note that our Meta-RL
approach shows greatly improved sample efficiency.

\paragraph{Environments: Bilateral Trade on Atari 2600.}
As a second, significantly higher-dimensional and challenging domain, we present a bilateral trade scenario on a modified Atari 2600 game (a state-of-the-art  domain in single-agent RL). We use a two-player version of  {\em Space Invaders}, and introduce an artificial resource constraint: Each agent can only fire in the game if they have a bullet available. Initially, neither player has any bullets available. Throughout the episode, we give bullets to player 1, one at a time at stochastic intervals. Player 1 can then choose to offer the sell this bullet to player 2 by offering them a price, or Player 1 can choose to use the bullet themselves. Player 2 in turn can choose to accept or reject a particular offer at a particular price. If a trade takes place, the sales price is added to player 1's reward, and subtracted from player 2's reward. Additionally, we introduce a reward scale imbalance: Each time player 1 successfully shoots an alien invader, they get a reward of 0.1. However each time player 2 shoots an alien, they get a much higher reward of 1.0. Noting that even well-trained AI agents do not hit every single shot they take, we should still expect that player 2 be able to generate just under 1.0 reward from each bullet they fire, and player 1 a much smaller reward of just under 0.1.

\begin{figure*}[h!]
\begin{center}
\includegraphics[width=0.85\textwidth]{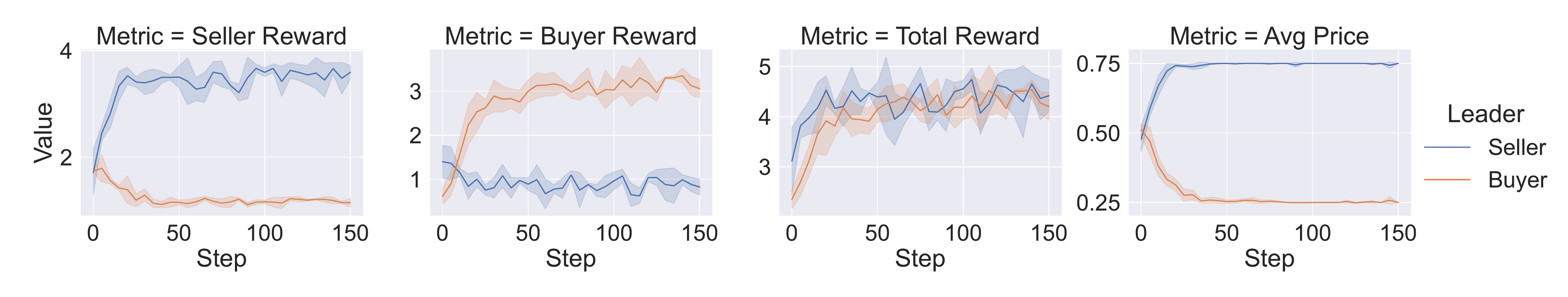}
\end{center}
\caption{Performance and behavior of PPO+Meta-RL on the Atari 2600 bilateral trade scenario. Plots show results for two distinct Stackelberg equilibria: Agent 1 (seller) as leader (blue curves) and agent 2 (buyer) as leader (orange). \label{fig:atari}}
\end{figure*}

There is more total reward generated if player 1 sells all their bullets to player 2, with the difference 
the ``gains from trade'' in economics. However,  this is not a mechanism design setting (there is no mechanism), and also that there are two Stackelberg equilibria: If player 1 is the leader, then their optimal strategy is to offer bullets to player 2 at just under player 2's average utility per bullet. Player 2 will best respond by accepting the trade, still generate small positive reward, and player 1 will receive almost the entirety of the gains from trade. In the second Stackelberg equilibrium, player 2 is the leader. Player 2's optimal strategy is to refuse any price higher than just above player 1's average utility per bullet; and player 1's best response is to offer to sell at that (low) price. In this scenario, player 1 will be left with little more reward than had they kept and used the bullets themselves, and player 2 will receive almost all the gains from trade.

Figure~\ref{fig:atari} shows that the Meta-RL algorithm is able to successfully learn this optimal behavior
for both equilibria. In this experiment we use discrete prices ($0, 0.25, 0.5, 0.75. 1.0$) for compatibility with the discrete Atari environment, so the results shown are the exact optimum.

\section{Conclusion}

We have introduced a general framework for using multi-agent RL approaches to find Stackelberg equilibria in Markov games, and discussed how this encompasses several approaches in the literature, while also conveying a much larger design space. 
{In addition, we show the necessity of our POMDP construction
 by showing that RL against} {followers that immediately best respond} 
{can fail, an important and surprising result by itself.}
 As a second key contribution, we have proposed and evaluated a novel approach to Stackelberg learning that uses Meta-RL to implement the follower oracle. This shows the power of  Theorem~\ref{thm:main},
which enables this approach, and is also a key contribution itself.
Our approach matches or exceeds the final mean episode reward of previous approaches, and does so 
at greatly improved speed of convergence. It also enables Stackelberg learning in  domains beyond the reach of previous approaches, which we show for a novel Atari 2600-based bilateral trade scenario. Finally, we show theoretically and experimentally the limits of Theorem~\ref{thm:main}, and in particular that RL algorithms can provably be unable to learn without the query-oracle special case construction. 

In addition to the technical results, we would like to offer a more high-level interpretation of the framework. A useful way to think about learning Stackelberg equilibria in Markov games is that they are, in a way, two problems in one: One, how does my strategy, i.e., choice of policy, affect the best-response of other agents? Two, how does my interactions with the environment, i.e., actions at each step, affect the reward I (and others) get? These are two very different problems, even operating at different levels---entire policy, versus action at each step. Theorem~\ref{thm:main} is giving a way to reconcile the best-response ``meta-level'' and the environment-interaction ``RL problem." In the general case, using techniques such as direct gradient descent or evolutionary policies, we focus on the best-response meta-level and either ignore the environment interaction (in evolutionary strategies) or subsume them inside an end-to-end differentiation (in direct policy gradient). 
In contrast, in the query-oracle special case, we focus on the environment interaction RL problem, and implicitly work the follower best-response into this. One way of looking at the contextual-policy follower oracle is that it makes the latter more feasible, by greatly reducing the number of leader queries compared to real environment interaction. 

We hope that this Meta-RL approach will enable Stackelberg RL approaches to scale up to richer settings, both through the explicit, contextual-policy approach taken in this paper, as well as approaches that infer context through recurrent networks. %
Beyond this, we hope the framework of Theorem~\ref{thm:main} will inspire novel ways of thinking about Stackelberg RL. One potential avenue for future work 
is to study approaches that explicitly take into account both the ``meta-level'' and ``environment-interaction'' problems. We believe that doing so could enable Stackelberg RL to scale to much more complex scenarios, and open novel applications. If successful,  this may enable
the automated, {end-to-end}
learning of system design beyond traditional settings
such as security games and mechanism design.

\section*{Acknowledgments}
{The project was sponsored, in part, by a grant from 
the Cooperative AI Foundation. 
The content 
does
not necessarily reflect the position or the policy of the
Cooperative AI Foundation
 and no  endorsement should be inferred.}

\bibliography{stackelberg}

\begin{thebibliography}{41}
\providecommand{\natexlab}[1]{#1}
\providecommand{\url}[1]{\texttt{#1}}
\expandafter\ifx\csname urlstyle\endcsname\relax
  \providecommand{\doi}[1]{doi: #1}\else
  \providecommand{\doi}{doi: \begingroup \urlstyle{rm}\Url}\fi

\bibitem[An et~al.(2017)An, Tambe, and Sinha]{an2017stackelberg}
An, B., Tambe, M., and Sinha, A.
\newblock Stackelberg security games {(SSG)} basics and application overview.
\newblock \emph{Improving Homeland Security Decisions}, pp.\  485, 2017.

\bibitem[Bai et~al.(2021)Bai, Jin, Wang, and Xiong]{bai2021sample}
Bai, Y., Jin, C., Wang, H., and Xiong, C.
\newblock Sample-efficient learning of {Stackelberg equilibria} in general-sum
  games.
\newblock \emph{Advances in Neural Information Processing Systems},
  34:\penalty0 25799--25811, 2021.

\bibitem[Balaguer et~al.(2022)Balaguer, Koster, Summerfield, and
  Tacchetti]{balaguer2022good}
Balaguer, J., Koster, R., Summerfield, C., and Tacchetti, A.
\newblock The good shepherd: An oracle agent for mechanism design.
\newblock \emph{arXiv preprint arXiv:2202.10135}, 2022.

\bibitem[Bergstra et~al.(2013)Bergstra, Yamins, and Cox]{bergstra2013making}
Bergstra, J., Yamins, D., and Cox, D.
\newblock Making a science of model search: Hyperparameter optimization in
  hundreds of dimensions for vision architectures.
\newblock In \emph{International conference on machine learning}, pp.\
  115--123. PMLR, 2013.

\bibitem[Blum et~al.(2014)Blum, Haghtalab, and Procaccia]{blum2014learning}
Blum, A., Haghtalab, N., and Procaccia, A.~D.
\newblock Learning optimal commitment to overcome insecurity.
\newblock \emph{Advances in Neural Information Processing Systems}, 27, 2014.

\bibitem[Brero et~al.(2021{\natexlab{a}})Brero, Chakrabarti, Eden,
  Gerstgrasser, Li, and Parkes]{brero2021learning}
Brero, G., Chakrabarti, D., Eden, A., Gerstgrasser, M., Li, V., and Parkes,
  D.~C.
\newblock Learning {Stackelberg} equilibria in sequential price mechanisms.
\newblock In \emph{Proc. ICML Workshop for Reinforcement Learning Theory},
  2021{\natexlab{a}}.

\bibitem[Brero et~al.(2021{\natexlab{b}})Brero, Eden, Gerstgrasser, Parkes, and
  Rheingans-Yoo]{brero2021reinforcement}
Brero, G., Eden, A., Gerstgrasser, M., Parkes, D., and Rheingans-Yoo, D.
\newblock Reinforcement learning of sequential price mechanisms.
\newblock In \emph{Proceedings of the AAAI Conference on Artificial
  Intelligence}, pp.\  5219--5227, 2021{\natexlab{b}}.

\bibitem[Brero et~al.(2022)Brero, Lepore, Mibuari, and
  Parkes]{brero2022learning}
Brero, G., Lepore, N., Mibuari, E., and Parkes, D.~C.
\newblock Learning to mitigate ai collusion on economic platforms.
\newblock \emph{Advances in Neural Information Processing Systems}, 35, 2022.

\bibitem[Duan et~al.(2016)Duan, Schulman, Chen, Bartlett, Sutskever, and
  Abbeel]{duan2016rl}
Duan, Y., Schulman, J., Chen, X., Bartlett, P.~L., Sutskever, I., and Abbeel,
  P.
\newblock Rl$^2$: Fast reinforcement learning via slow reinforcement learning.
\newblock \emph{arXiv preprint arXiv:1611.02779}, 2016.

\bibitem[Foerster et~al.(2017)Foerster, Chen, Al-Shedivat, Whiteson, Abbeel,
  and Mordatch]{foerster2017learning}
Foerster, J.~N., Chen, R.~Y., Al-Shedivat, M., Whiteson, S., Abbeel, P., and
  Mordatch, I.
\newblock Learning with opponent-learning awareness.
\newblock \emph{arXiv preprint arXiv:1709.04326}, 2017.

\bibitem[Humplik et~al.(2019)Humplik, Galashov, Hasenclever, Ortega, Teh, and
  Heess]{humplik2019meta}
Humplik, J., Galashov, A., Hasenclever, L., Ortega, P.~A., Teh, Y.~W., and
  Heess, N.
\newblock Meta reinforcement learning as task inference.
\newblock \emph{arXiv preprint arXiv:1905.06424}, 2019.

\bibitem[Letchford et~al.(2009)Letchford, Conitzer, and
  Munagala]{letchford2009learning}
Letchford, J., Conitzer, V., and Munagala, K.
\newblock Learning and approximating the optimal strategy to commit to.
\newblock In \emph{International symposium on algorithmic game theory}, pp.\
  250--262. Springer, 2009.

\bibitem[Li et~al.(2022)Li, Jia, Mate, Jabbari, Chakraborty, Tambe, and
  Vorobeychik]{li2022solving}
Li, Z., Jia, F., Mate, A., Jabbari, S., Chakraborty, M., Tambe, M., and
  Vorobeychik, Y.
\newblock Solving structured hierarchical games using differential backward
  induction.
\newblock In \emph{Uncertainty in Artificial Intelligence}, pp.\  1107--1117.
  PMLR, 2022.

\bibitem[Liang et~al.(2018)Liang, Liaw, Nishihara, Moritz, Fox, Goldberg,
  Gonzalez, Jordan, and Stoica]{liang2018rllib}
Liang, E., Liaw, R., Nishihara, R., Moritz, P., Fox, R., Goldberg, K.,
  Gonzalez, J.~E., Jordan, M.~I., and Stoica, I.
\newblock {RLlib}: Abstractions for distributed reinforcement learning.
\newblock In \emph{International Conference on Machine Learning ({ICML})},
  2018.

\bibitem[Liu(1998)]{liu1998stackelberg}
Liu, B.
\newblock {Stackelberg-Nash} equilibrium for multilevel programming with
  multiple followers using genetic algorithms.
\newblock \emph{Computers \& Mathematics with Applications}, 36\penalty0
  (7):\penalty0 79--89, 1998.

\bibitem[Lu et~al.(2022)Lu, Willi, De~Witt, and Foerster]{lu2022model}
Lu, C., Willi, T., De~Witt, C. A.~S., and Foerster, J.
\newblock Model-free opponent shaping.
\newblock In \emph{International Conference on Machine Learning}, pp.\
  14398--14411. PMLR, 2022.

\bibitem[Mishra et~al.(2017)Mishra, Rohaninejad, Chen, and
  Abbeel]{mishra2017simple}
Mishra, N., Rohaninejad, M., Chen, X., and Abbeel, P.
\newblock A simple neural attentive meta-learner.
\newblock \emph{arXiv preprint arXiv:1707.03141}, 2017.

\bibitem[Mnih et~al.(2013)Mnih, Kavukcuoglu, Silver, Graves, Antonoglou,
  Wierstra, and Riedmiller]{mnih2013playing}
Mnih, V., Kavukcuoglu, K., Silver, D., Graves, A., Antonoglou, I., Wierstra,
  D., and Riedmiller, M.
\newblock Playing atari with deep reinforcement learning.
\newblock \emph{arXiv preprint arXiv:1312.5602}, 2013.

\bibitem[Mnih et~al.(2015)Mnih, Kavukcuoglu, Silver, Rusu, Veness, Bellemare,
  Graves, Riedmiller, Fidjeland, Ostrovski, et~al.]{mnih2015human}
Mnih, V., Kavukcuoglu, K., Silver, D., Rusu, A.~A., Veness, J., Bellemare,
  M.~G., Graves, A., Riedmiller, M., Fidjeland, A.~K., Ostrovski, G., et~al.
\newblock Human-level control through deep reinforcement learning.
\newblock \emph{nature}, 518\penalty0 (7540):\penalty0 529--533, 2015.

\bibitem[Nakamura(2015)]{nakamura2015one}
Nakamura, T.
\newblock One-leader and multiple-follower {Stackelberg games} with private
  information.
\newblock \emph{Economics Letters}, 127:\penalty0 27--30, 2015.

\bibitem[Nisan \& Ronen(1999)Nisan and Ronen]{nisan1999algorithmic}
Nisan, N. and Ronen, A.
\newblock Algorithmic mechanism design.
\newblock In \emph{Proceedings of the thirty-first annual ACM symposium on
  Theory of computing}, pp.\  129--140, 1999.

\bibitem[Paruchuri et~al.(2008)Paruchuri, Pearce, Marecki, Tambe, Ordonez, and
  Kraus]{paruchuri2008playing}
Paruchuri, P., Pearce, J.~P., Marecki, J., Tambe, M., Ordonez, F., and Kraus,
  S.
\newblock Playing games for security: {An} efficient exact algorithm for
  solving {Bayesian Stackelberg} games.
\newblock In \emph{Proceedings of the 7th international joint conference on
  Autonomous agents and multiagent systems-Volume 2}, pp.\  895--902, 2008.

\bibitem[Peng et~al.(2019)Peng, Shen, Tang, and Zuo]{peng2019learning}
Peng, B., Shen, W., Tang, P., and Zuo, S.
\newblock Learning optimal strategies to commit to.
\newblock In \emph{Proceedings of the AAAI Conference on Artificial
  Intelligence}, pp.\  2149--2156, 2019.

\bibitem[Rakelly et~al.(2019)Rakelly, Zhou, Finn, Levine, and
  Quillen]{rakelly2019efficient}
Rakelly, K., Zhou, A., Finn, C., Levine, S., and Quillen, D.
\newblock Efficient off-policy meta-reinforcement learning via probabilistic
  context variables.
\newblock In \emph{International conference on machine learning}, pp.\
  5331--5340. PMLR, 2019.

\bibitem[Robinson \& Goforth(2005)Robinson and Goforth]{robinson2005topology}
Robinson, D. and Goforth, D.
\newblock \emph{The topology of the 2x2 games: a new periodic table}, volume~3.
\newblock Psychology Press, 2005.

\bibitem[Schulman et~al.(2017)Schulman, Wolski, Dhariwal, Radford, and
  Klimov]{schulman2017proximal}
Schulman, J., Wolski, F., Dhariwal, P., Radford, A., and Klimov, O.
\newblock Proximal policy optimization algorithms.
\newblock \emph{arXiv preprint arXiv:1707.06347}, 2017.

\bibitem[Shi et~al.(2019)Shi, Yu, Wang, Wang, Zhang, Lai, and
  An]{shi2019learning}
Shi, Z., Yu, R., Wang, X., Wang, R., Zhang, Y., Lai, H., and An, B.
\newblock Learning expensive coordination: An event-based deep rl approach.
\newblock In \emph{International Conference on Learning Representations}, 2019.

\bibitem[Shu \& Tian(2018)Shu and Tian]{shu2018m}
Shu, T. and Tian, Y.
\newblock M$^3$rl: Mind-aware multi-agent management reinforcement learning.
\newblock \emph{arXiv preprint arXiv:1810.00147}, 2018.

\bibitem[Sinha et~al.(2014)Sinha, Malo, Frantsev, and Deb]{sinha2014finding}
Sinha, A., Malo, P., Frantsev, A., and Deb, K.
\newblock Finding optimal strategies in a multi-period multi-leader--follower
  {Stackelberg} game using an evolutionary algorithm.
\newblock \emph{Computers \& Operations Research}, 41:\penalty0 374--385, 2014.

\bibitem[Sinha et~al.(2018)Sinha, Fang, An, Kiekintveld, and
  Tambe]{sinha2018stackelberg}
Sinha, A., Fang, F., An, B., Kiekintveld, C., and Tambe, M.
\newblock Stackelberg security games: {Looking} beyond a decade of success.
\newblock In \emph{Proceedings of the Twenty-Seventh International Joint
  Conference on Artificial Intelligence (IJCAI-18)}, 2018.

\bibitem[Solis et~al.(2016)Solis, Clempner, and Poznyak]{solis2016modeling}
Solis, C.~U., Clempner, J.~B., and Poznyak, A.~S.
\newblock Modeling multileader--follower noncooperative {Stackelberg} games.
\newblock \emph{Cybernetics and Systems}, 47\penalty0 (8):\penalty0 650--673,
  2016.

\bibitem[Sutton et~al.(1999)Sutton, McAllester, Singh, and
  Mansour]{sutton1999policy}
Sutton, R.~S., McAllester, D., Singh, S., and Mansour, Y.
\newblock Policy gradient methods for reinforcement learning with function
  approximation.
\newblock \emph{Advances in neural information processing systems}, 12, 1999.

\bibitem[Swamy(2007)]{swamy2007effectiveness}
Swamy, C.
\newblock The effectiveness of {Stackelberg} strategies and tolls for network
  congestion games.
\newblock In \emph{SODA}, pp.\  1133--1142. Citeseer, 2007.

\bibitem[Wang et~al.(2016)Wang, Kurth-Nelson, Tirumala, Soyer, Leibo, Munos,
  Blundell, Kumaran, and Botvinick]{wang2016learning}
Wang, J.~X., Kurth-Nelson, Z., Tirumala, D., Soyer, H., Leibo, J.~Z., Munos,
  R., Blundell, C., Kumaran, D., and Botvinick, M.
\newblock Learning to reinforcement learn.
\newblock \emph{arXiv preprint arXiv:1611.05763}, 2016.

\bibitem[Wang et~al.(2022)Wang, Xu, Perrault, Reiter, and
  Tambe]{wang2022coordinating}
Wang, K., Xu, L., Perrault, A., Reiter, M.~K., and Tambe, M.
\newblock Coordinating followers to reach better equilibria: End-to-end
  gradient descent for {stackelberg} games.
\newblock In \emph{Proceedings of the AAAI Conference on Artificial
  Intelligence}, pp.\  5219--5227, 2022.

\bibitem[Xu et~al.(2014)Xu, Fang, Jiang, Conitzer, Dughmi, and
  Tambe]{xu2014computing}
Xu, H., Fang, F., Jiang, A.~X., Conitzer, V., Dughmi, S., and Tambe, M.
\newblock Computing minimax strategy for discretized spatio-temporal zero-sum
  security games.
\newblock In \emph{International Joint Workshop on Optimization in Multi-Agent
  Systems and Distributed Constraint Reasoning (OPTMASDCR) at AAMAS}. Citeseer,
  2014.

\bibitem[Yang et~al.(2022)Yang, Wang, Trivedi, Zhao, and Zha]{yang2022adaptive}
Yang, J., Wang, E., Trivedi, R., Zhao, T., and Zha, H.
\newblock Adaptive incentive design with multi-agent meta-gradient
  reinforcement learning.
\newblock In \emph{Proceedings of the 21st International Conference on
  Autonomous Agents and Multiagent Systems}, pp.\  1436--1445, 2022.

\bibitem[Zhang et~al.(2016)Zhang, Xiao, Cai, Niyato, Song, and
  Han]{zhang2016multi}
Zhang, H., Xiao, Y., Cai, L.~X., Niyato, D., Song, L., and Han, Z.
\newblock A multi-leader multi-follower {Stackelberg game} for resource
  management in {LTE} unlicensed.
\newblock \emph{IEEE Transactions on Wireless Communications}, 16\penalty0
  (1):\penalty0 348--361, 2016.

\bibitem[Zheng et~al.(2022)Zheng, Trott, Srinivasa, Parkes, and
  Socher]{zheng2022ai}
Zheng, S., Trott, A., Srinivasa, S., Parkes, D.~C., and Socher, R.
\newblock The {AI Economist: Taxation} policy design via two-level deep
  multiagent reinforcement learning.
\newblock \emph{Science Advances}, 8\penalty0 (18):\penalty0 eabk2607, 2022.

\bibitem[Zhong et~al.(2021)Zhong, Yang, Wang, and Jordan]{zhong2021can}
Zhong, H., Yang, Z., Wang, Z., and Jordan, M.~I.
\newblock Can reinforcement learning find {Stackelberg-Nash} equilibria in
  general-sum markov games with myopic followers?
\newblock \emph{arXiv preprint arXiv:2112.13521}, 2021.

\bibitem[Zintgraf et~al.(2019)Zintgraf, Igl, Shiarlis, Mahajan, Hofmann, and
  Whiteson]{zintgraf2019variational}
Zintgraf, L., Igl, M., Shiarlis, K., Mahajan, A., Hofmann, K., and Whiteson, S.
\newblock Variational task embeddings for fast adapta-tion in deep
  reinforcement learning.
\newblock In \emph{International Conference on Learning Representations
  Workshop (ICLRW)}, 2019.

\end{thebibliography}
\bibliographystyle{icml2023}

\newpage
\appendix
\onecolumn

\section{Proof and Discussion of Theorem~\ref{thm:main}}
\label{sec:ap:proof}
We include here the proof of Lemma~\ref{thm:lemma} and Theorem~\ref{thm:main}.\setcounter{theorem}{0}
\setcounter{lemma}{0}
\begin{lemma}
Given a Markov Game $\MDP$ and a follower equilibrium oracle $\mathcal E$, let $\mathcal L_\MDP$ be the learning problem the leader faces. If:
\begin{enumerate}
    \item for each choice of leader policy $\pi_L$, $\mathcal L$ computes the follower best-response $\mathcal E(\pi_L)$, and \label{cond:equi}
    \item $\mathcal L(\pi_L)$ evaluates the leader policy $\pi_L$ against the follower best-response $\mathcal E(\pi_L)$ in $\MDP$, i.e. the value of $\mathcal L(\pi_L)$ is $r_{L}(\pi_L, \mathcal E(\pi_L))$ in $\MDP$, \label{cond:reward}
\end{enumerate}
then an optimal solution $\pi_L^*$ to $\mathcal L$ together with the follower best-response $\mathcal E(\pi_L^*)$ form a Stackelberg equilibrium in $\MDP$.
\end{lemma}

\begin{proof} 
Assume $s^*_L$ optimally solves $\mathcal L$, i.e. $r_{L,
\mathcal L}(s^*_L) = \max r_{L, \mathcal L}(\pi_L)$. 
By condition \ref{cond:reward}, the leader's reward in $\mathcal L$ is the same as that in $\mathcal M$ when the followers play their best-response equilibrium, i.e. $r_{L,
\mathcal L}(s^*_L) = \max r_{L, \mathcal L}(\pi_L) = \max r_{L, \mathcal M}(\pi_L, \mathcal E(\pi_L))$. This immediately means that $s^*_L$ together with $\mathcal E(\pi_L)$ form a Stackelberg equilibrium in $\mathcal M$. Condition \ref{cond:equi} is only required implicitly to ensure that followers are playing their best-response equilibrium when the leader strategy $\pi_L$ is evaluated in $\MDP$. This shows the general case.
\end{proof}

We now show the main theorem.

\begin{theorem}
Given a Markov Game $\MDP$,
and a follower equilibrium oracle $\mathcal E$, if in addition to the conditions of Lemma~\ref{thm:lemma}, the follower oracle $\mathcal E$ is a \textit{query oracle} (Definition~\ref{def:queryoracle}), then the leader learning problem $\mathcal L$ can be constructed as a POMDP.
\end{theorem}

\begin{proof} 
Given a Markov Game $\MDP$ and a follower best-response oracle $\mathcal E$ that only requires query access to the leader strategy $\pi_L$, define {\em a new leader POMDP $\mathcal L$} as follows: The action and observation space for the leader in $\mathcal L$ are the same as those in $\MDP$. $\mathcal L$ is then constructed in two parts. First, the leader is queried by the oracle $\mathcal R$; then an episode from the original Markov game $\MDP$ plays out:
\begin{itemize}
    \item \textbf{Initial Segment:} For an initial number of steps in $\mathcal L$, each step performs one query from the follower oracle $\mathcal E$: If a given query wishes to determine the leader policy's response to observation $o$, then the leader will receive $o$ as its observation in $\mathcal L$, and the leader's action will be given to $\mathcal E$ as the response to its query. The leader will receive no reward in these steps. 
    \item \textbf{Final Segment:} Once a follower equilibrium $\mathbf \pi_F$ has been determined, the remainder of $\mathcal L$ will be constructed from the original Markov game $\mathcal M$: We let followers act according to the computed follower equilibrium $\mathbf \pi_F$ and treat them (including all their internal state) as part of the environment. 
\end{itemize}
We now show that $\mathcal L$ is a POMDP. 
 
\textbf{POMDP, setup:} Let the state of $\mathcal L$ be $z_t = (z_{\mathcal E,t}, z_{\mathcal M,t}, z_{F,t})$, the internal state of the follower equilibrium oracle (in the initial part of the $\mathcal L$), the state of the original Markov Game, and the internal state of the follower agents (in the final part of the $\mathcal L$). In the initial part wlog assume this is $(z_{\mathcal E}, 0,0)$, and in the final part $(0, z_{\mathcal M,t}, z_{F,t})$.

\textbf{POMDP, part 1:} By assumption, $\mathcal E$ only requires query access to $\pi_L$, i.e. if at timestep $t$, the oracle's internal state is $z_{\mathcal E, t}$ and the oracle issues the query $o_t$, then the oracle's next internal state $z_{\mathcal E, t+1}$ is a function of only $z_{\mathcal E, t}$ and $q_t$, the leader's response to the query $o_t$.
By the construction of the first part of $\mathcal L$, we have that the leader's observation at timestep $t$ is precisely the oracle query $o_t$, and so it's action $a_t$ gives the oracle response $q_t$. Together, we get that the $\mathcal L$ state at time $t+1$, $z_{t+1}$, is a function of only $z_t$ and $a_t$, showing that $\mathcal L$ is a POMDP in the initial part.

\textbf{POMDP, part 2:} In the final part of $\mathcal L$, at step $t$, the Markov Game state is $z_{\mathcal M,t}$, and leader and follower observations depend only on this state, i.e. $o_{L,t} = o_{L,t}(z_{\mathcal M,t})$ and $o_{F,t} = o_{F,t}(z_{\mathcal M,t})$. In turn, both the follower actions $a_{F,t}$ as well as the next follower state $z_{F,t+1}$, only depend on $o_{F,t}$ and the current follower state $z_{F,t}$; therefore both depend only on $z_{\mathcal M,t}$ and $z_{F,t}$. In turn, the next state of $\mathcal M$, $z_{\mathcal M,t+1}$, depends on leader and follower actions, and therefore only on leader action, $z_{\mathcal M,t}$ and $z_{F,t}$. Together, it follows that $z_{t+1}$ only depends on $z_{t}$ and the leader's action $a_{L,t}$, meaning the final part of $\mathcal L$ is Markovian.

We have therefore shown that $\mathcal L$ as a whole is a POMDP. We now show that an optimal policy in $\mathcal L$ forms a Stackelberg equilibrium.

\textbf{Stackelberg:}
By the assumption that the leader policy is invariant, we have that if $\pi_L(o) = a$ in response to an oracle query, then $\pi_L(o) = a$ in the Markov Game $\mathcal M$ as well. Therefore, the follower equilibrium $\mathbf \pi_F$ computed by the oracle at the end of the initial part of $\mathcal L$ is a best-response equilibrium to the strategy the leader plays in $\mathcal M$ in the final part of $\mathcal L$. 

Now, by construction of $\mathcal L$, the leader reward given any $\pi_L$ in $\mathcal L$ is the same as the leader reward in $\mathcal M$ when followers play $\mathbf{\pi}_F$, and by the above $\mathbf \pi_F$ is indeed a best-response equilibrium, i.e. $r_{L, \mathcal L}(\pi_L) = r_{L,\mathcal M}(\pi_L, \mathbf \pi_F) = r_{L,\mathcal M}(\pi_L, \mathcal E(\pi_L))$. 
Finally, by optimality of $\pi_L^*$ in $\mathcal L$, $\pi_L^* \in \mathrm{argmax}(r_{L, \mathcal L}(\pi_L))$, and therefore $\pi_L^* \in \mathrm{argmax}(r_{L,\mathcal M}(\pi_L, \mathcal E(\pi_L))) $. But this precisely means that $\pi_L^*$ and $\mathcal E(\pi_L^*)$ form a Stackelberg equilibrium in $\mathcal M$.
\end{proof}

\paragraph{Discussion}
We intentionally stated the lemma and theorem in a fairly abstract manner, so as to be general and cover a wide range of possible oracle implementations. The theorem may be more readily understood through concrete examples:

In the simplest case, the follower oracle is implemented using \textbf{reinforcement learning}, i.e. the leader and follower(s) all use RL. In this scenario, the initial segment of $\mathcal L$ is simply one or more episodes of $\MDP$, where the followers are learning. The final segment is an episode of $\MDP$ when the followers have converged and are not learning anymore. This ``looks'' very similar to a standard independent-learning multi-agent RL setup, but with some crucial differences (necessary due to conditions 2 and 3 from the general case of the theorem): For the leader, the initial and final segment form one single episode (but are treated as multiple episodes for the followers), and the leader does not receive reward in the initial-segment episodes. A variant of this is done in \cite{brero2022learning}, and variations that do not strictly follow conditions 2 and 3 (and thus do not strictly guarantee Stackelberg) are common in the literature as discussed in section \ref{sec:contextual}.

A related case is a follower oracle implemented using a different learning approach, such as \textbf{no-regret dynamics} in \cite{brero2021learning}. In this case, the leader POMDP looks similar to the RL case, except in the initial segment the followers are now learning using a no-regret algorithm such as multiplicative weights. This is qualitatively different, as these algorithms explore in a more systematic way than RL algorithms do.

Note that for the purposes of Theorem~\ref{thm:main} and in contrast to typical multi-agent RL setups, in both the above situations we view the leader learning as separate from the follower (reinforcement or no-regret) learning. The latter in our model is an algorithm to implement the follower best-response oracle, and it is useful to think of it in this way. In this view, the initial segment isn't a joint leader-follower multi-agent system, it is the follower oracle algorithm querying the leader policy. In the above cases the follower oracle happens to use the same or a similar algorithm as we use to learn the leader behavior, but a crucial consequence of our statement of Definition~\ref{def:stackelberg} and Theorem~\ref{thm:main} is that this need not be the case. Indeed, for any follower oracle algorithm that only uses query access to the leader policy, we can use the same construction. For instance, in our \textbf{Meta-RL} approach, we use a fixed set of queries $o_0, \ldots, o_k$ to define the context for the meta-follower. In this case, the initial segment in the leader POMDP $\mathcal L$ will always be the same sequence of observations $o_0, \ldots, o_k$. The leader's actions in responses to these observations in turn are used to form the context for the meta-follower in the final segment. The final segment is an episode of $\MDP$ played between the leader and the meta-follower, whose behavior is informed by the context.

Finally, what would a follower oracle implementation look like that falls outside the query-oracle special case of Theorem~\ref{thm:main} (but within the general case)? \textit{Firstly}, any of the above cases have a non-POMDP counterpart: We could simply omit the initial segment from the leader's training batch (but in a centralized training regime we could still run forward passes through the leader's neural network to answer the follower oracle queries). Crucially, the followers' behavior changes without any action being taken by the leader. This makes such a construction non-Markovian from the leader's point of view (thought the general case of Theorem~\ref{thm:main} still applies). However, this is not the most interesting case.
\textit{Secondly}, any follower oracle algorithm that makes use of a \textit{description} of the leader policy, rather than query access, would be incompatible with the POMDP construction. For instance, if the leader policy $\pi_L$ is parametrized by weights $\theta$, it is conceivable that a follower oracle algorithm could compute a best response directly from $\theta$. In such a case, it is not possible to roll out the best-response computation into the leader's experiences as we do through the initial segment of the special case POMDP construction. Without this initial segment, however, learning can provably fail, as Theorem~\ref{thm:divergence} shows

\section{Theorem~\ref{thm:main}: Leader Memory and Leader Invariance} \label{sec:ap:memory}
\paragraph{Leader Memory.} We state the query-oracle case of Theorem~\ref{thm:main} for memory-less leader policies, i.e. leader policies that map directly from observations to actions. This is without loss of generality because for leader policies that use memory we may take the view that the leader policy operates on belief states, mapping belief state to action. In this view, the theorem applies as-is, and we query the leader policy on belief states. This would work well, for instance, if leader memory was implemented through a sufficient statistic. Alternatively, if we want to treat memory as intrinsic to the leader policy, queries become sequences of observations. In this view, the proof applies {\em mutatis mutandis}. The main technicality in this case is to reset internal state of the leader policy between queries, so that queries are well-defined. This is also important in order to ensure the leader invariance conditions (an unrestricted LSTM could easily allow a leader to distinguish queries from real game).

\paragraph{Leader Invariance.} It may be possible to give the leader policy memory beyond the two cases above, i.e. memory with state that carries through between follower oracle queries and/or to real play. In any such cases, it is necessary that the leader policy be \textit{invariant}, meaning it is acting the same during the initial segment (i.e. oracle queries) and the final segment (i.e. original game) of the constructed leader POMDP $\mathcal L$. This has not been stated explicitly in previous works, but is a critical part of ensuring convergence to the correct equilibrium. If the leader policy were to act differently during the oracle queries, it could ``trick'' followers into suboptimal behavior that gives the leader better reward but is not a best-response, and thus not a Stackelberg equilibrium. For instance, in an iterated prisoners dilemma, a leader could pretend to be playing tit-for-tat during oracle queries, leading followers to cooperate; and  could then defect during the actual game. We show this experimentally in Appendix~\ref{sec:ap:limits:noninvariant}. Invariance is easily ensured if the leader policy cannot distinguish queries and real play, which is generally true for memory-less policies. Alternatively the leader policy could be explicitly constrained to be invariant, e.g. through an appropriate loss term.

\section{Limitations of Lemma~\ref{thm:lemma} and Theorem~\ref{thm:main}}
\label{sec:ap:limitations}
We now present experimental evidence of the limitations of our main Lemma and Theorem. In particular, we will show that violating any of the conditions of the theorem can lead to learning failure.

\subsection{Non-POMDP}
\label{sec:ap:limits:pomdp}
An interesting question we asked earlier is whether rolling out the follower queries into the leader episode to form a POMDP is strictly necessary. Theorem~\ref{thm:divergence} formally shows that this is the case, but we also test this experimentally here. Figure~\ref{fig:hiddenqueries} shows the performance of our approach on a slightly modified iterated prisoner's dilemma (see Appendix~\ref{sec:ap:divergence} for full payoff matrices). We show our standard setting where queries are part of the leader episode, as well as a setting where they are hidden from the leader. The hidden-queries setting fails to learn a sensible behavior. This is consistent across learning rates, and across algorithms. Note however that this only applies to RL algorithms. An approach that operates directly on the policy space such as Evolutionary Strategies is still able to learn successfully, as shown on the right hand side of Figure~\ref{fig:hiddenqueries}.
\begin{figure}
\begin{center}
\includegraphics[width=0.9\linewidth]{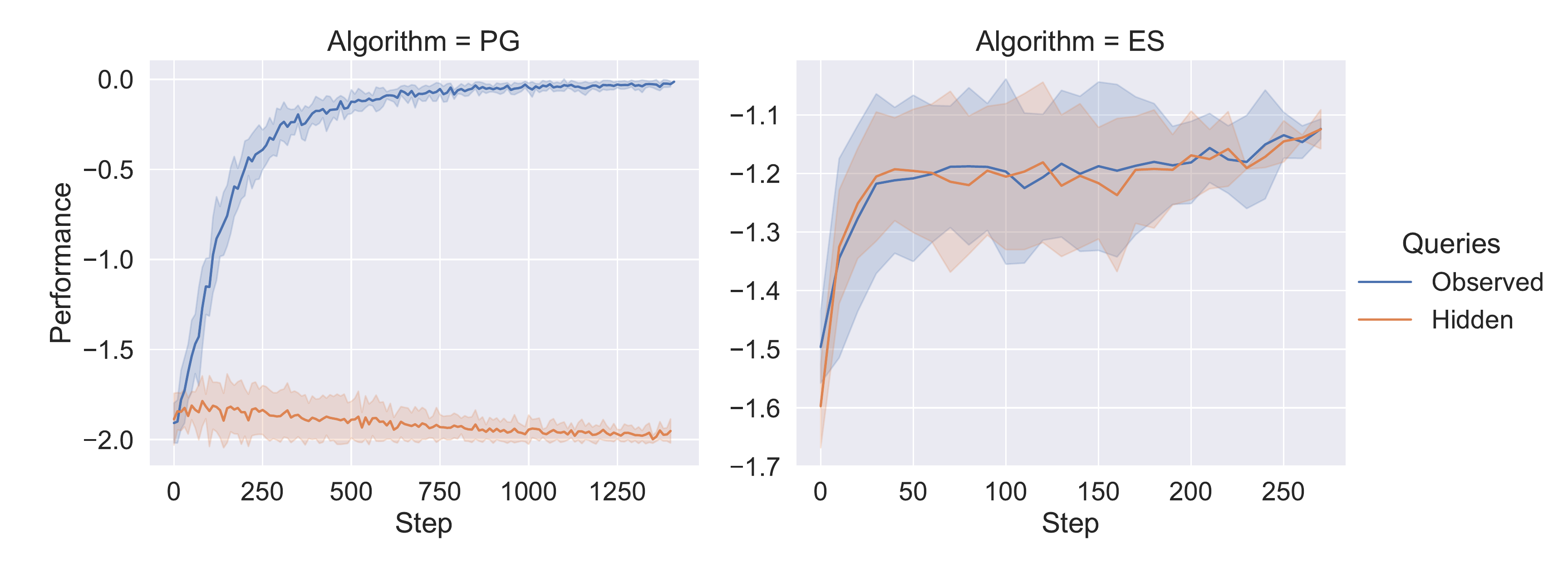}
\end{center}
\caption{Performance of contextual-policy approach with hidden and rolled-out oracle queries in iterated prisoners dilemma. \label{fig:hiddenqueries}}
\end{figure}

\newpage
    
\subsection{Non-Invariant Leader}
\label{sec:ap:limits:noninvariant}
\begin{wrapfigure}{}{0.5\textwidth}
\includegraphics[width=0.99\linewidth]{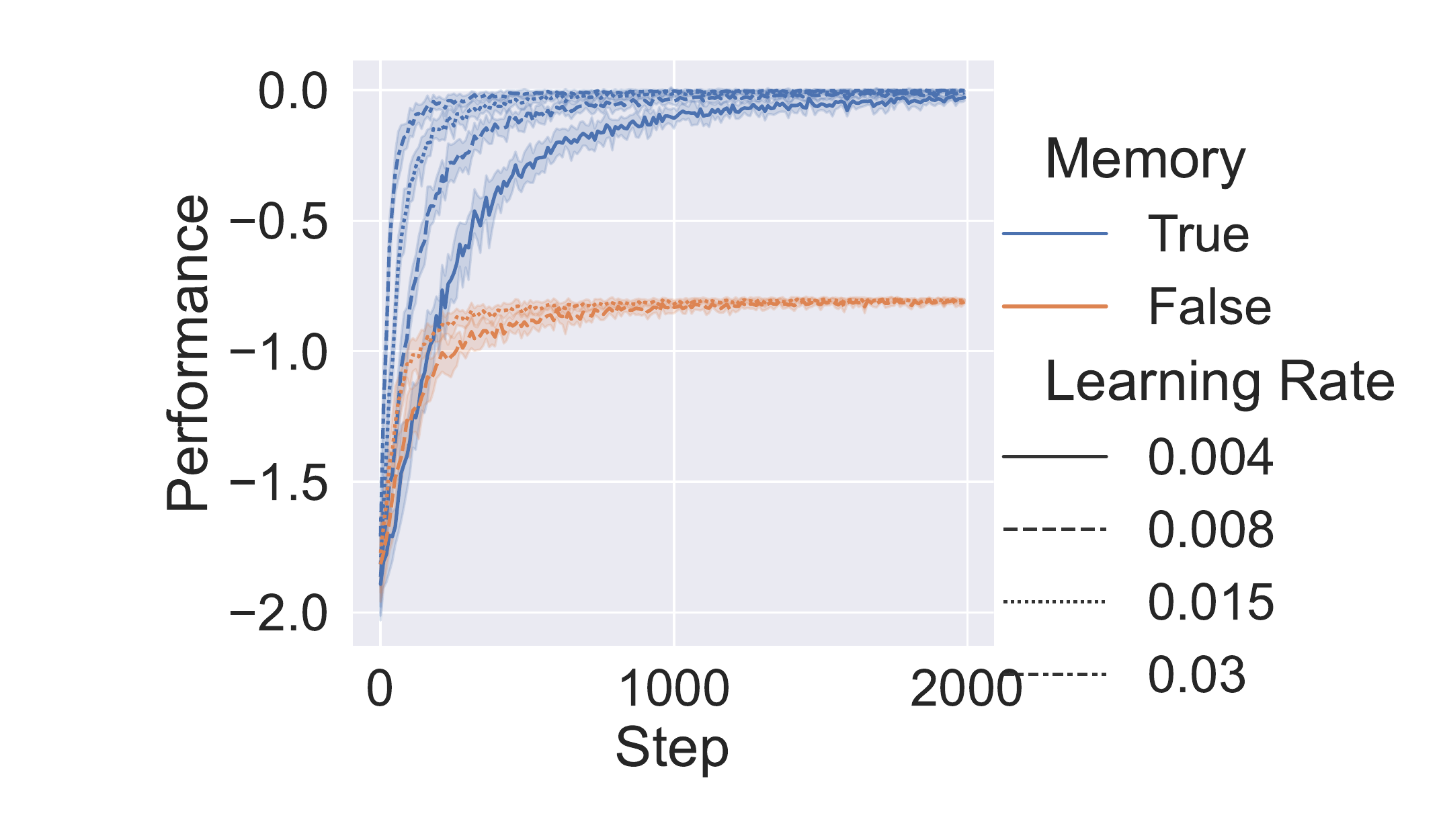} 
\caption{Leader reward with invariant and non-invariant policies in iterated prisoners dilemma.\label{fig:memory}}
\end{wrapfigure}
We also experimentally illustrate the leader invariance condition in Theorem~\ref{thm:main} in the iterated prisoner's dilemma setting. For simplicity, we emulate memory for the leader policy by concatenating a binary variable to its observation,  set to 0 during the first five steps of each episode (the queries), and 1 afterwards.\footnote{A neural network could learn to extract the same discriminator from a step counter, and a recurrent network could easily learn to keep such a counter.} 
As can be seen in Figure~\ref{fig:memory}, when given access to this additional variable, the leader gains significantly higher reward. The leader policy effectively learns to act as if it was playing tit-for-tat during the queries, thereby inducing the follower to respond by cooperating; the leader then always defects during the actual game, thereby achieving maximum reward. This is not a Stackelberg equilibrium.

\subsection{Leader Reward During Follower Learning}
\label{sec:ap:limits:reward}
\begin{figure}
\begin{center}
\includegraphics[width=0.9\linewidth]{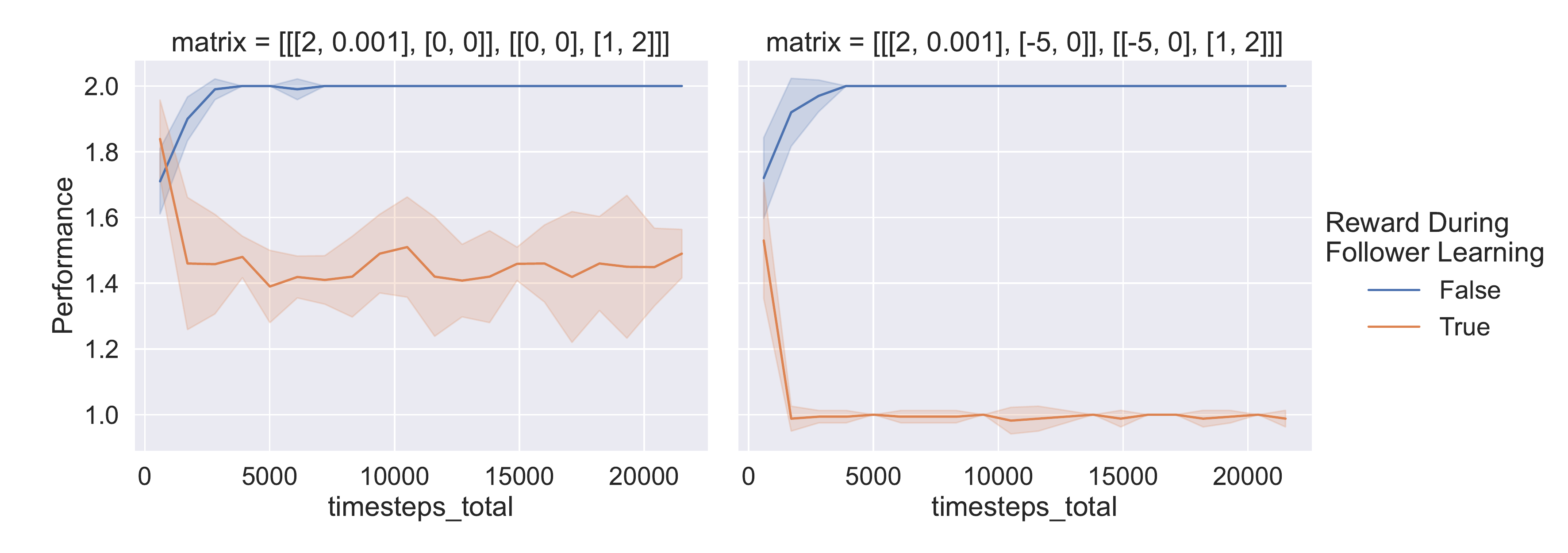}
\end{center}
\caption{Leader performance with and without reward during follower Q-learning. Clearly the leader fails to learn the sole Stackelberg equilibrium (reward 2.0) if reward is given during follower learning. Plots show reward during actual play only, i.e. without reward during follower Q-learning, as this is the relevant quantity for Stackelberg equilibria. \label{fig:bots_leaderreward}}
\end{figure}
One condition of Theorem~\ref{thm:main} is that the leader only be evaluated against followers who are best-responding. If the follower oracle is implemented using learning dynamics observable to the leader, this means that the leader must not receive reward during this learning phase. If the leader did receive reward, this could give the leader the wrong optimization target. Imagine for instance a setting where the leader has one strategy choice corresponding to a quickly-learnable follower best-response strategy that gives medium reward to the leader, and another leader strategy choice corresponding to a slow-to-learn follower strategy with high leader reward.
We can easily simulate this using a slightly modified version of the ``Battle of the Sexes'' single-shot matrix game we used as an example in the introduction. In this, we modify the follower reward so that the leader-preferred option gives the follower very little reward. We then couple this with carefully chosen (but entirely reasonable) Q-learning hyperparameters for the follower. As a result, a leader who receives reward during the follower learning phase is not able to reliable learn the correct equilibrium anymore, even in such a simple game, as Figure~\ref{fig:bots_leaderreward} shows. If we further modify the game to penalize the leader for coordination failure, this can even lead to the leader consistently learning the wrong coordination choice, as the right-hand plot shows.

Notice however that this (reward throughout follower learning) is also a valid target to optimize for, where the leader aims to optimize its expected return taking into account that followers may need some time to adjust to the leader's behavior. In the case of \cite{balaguer2022good} this is the intent, especially with regards to designing mechanisms for human participants as followers.

\subsection{Continuous Follower Learning}
\label{sec:ap:limits:continuous}
\begin{wrapfigure}{}{0.5\textwidth}
\includegraphics[width=0.9\linewidth]{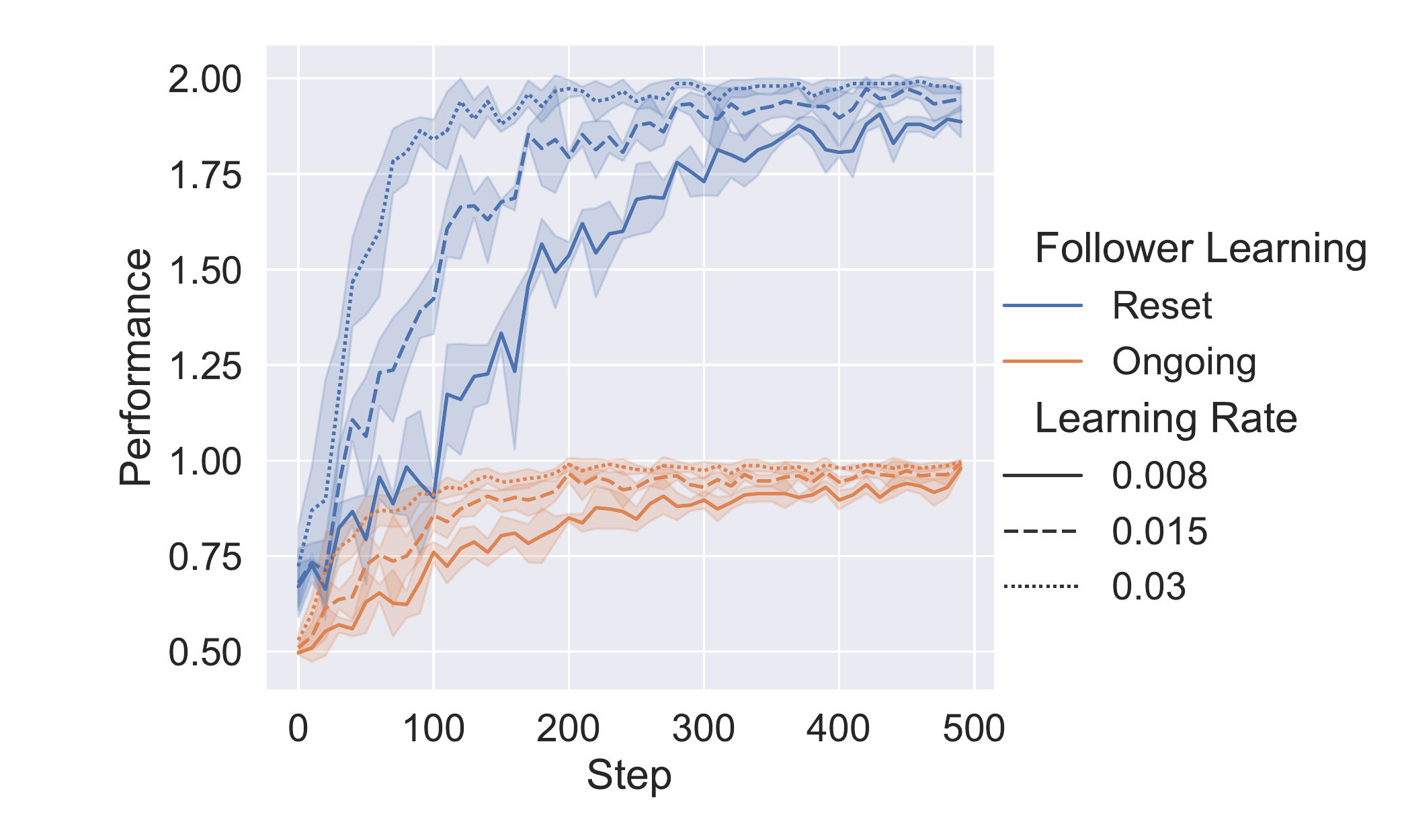} 
\caption{Leader reward with a Q-learning follower on Battle of the Sexes, where the follower initializes a blank Q-table each episode (blue) or keeps their previous Q-table (orange).\label{fig:reset}}
\end{wrapfigure}
Finally, virtually all previous approaches in the literature use some sort of learning dynamics to implement the follower oracle. A tempting way of improving learning speed in such a paradigm would be to retain follower policies between leader updates. That is, if at the end of the leader learning iteration $t$, the follower is best responding using strategy / policy parameters $\phi_{t,\text{end}}$, then instead of initializing follower weights $\phi_{t+1,\text{start}}$ randomly, set $\phi_{t+1,\text{start}} = \phi_{t,\text{end}}$.
Under the assumption that the leader policy only changed a little, and the conjecture that therefore the optimal follower policy only changed a little, this should allow follower learning to start from very near the optimum, and thus hopefully require much short inner (follower learning) loops. However, this has some drawbacks. For one, it makes the leader learning problem non-stationary. Beyond this, it can lead to learning failure, if both leader and follower get stuck on a local optimum. Figure~\ref{fig:reset} shows this in practice on the ``Battle of the Sexes'' example, where non-resetting follower learning can lead to convergence to the follower-preferred choice rather than the Stackelberg equilibrium.

\section{Proof of Non-POMDP Divergence}
\label{sec:ap:divergence}
We now present a proof of Theorem~\ref{thm:divergence}. A priori it is not clear that the query-oracle POMDP construction is strictly necessary, or if standard RL algorithm could also learn without it. Without the POMDP construction, the leader would effectively always play against followers who immediately best respond. The following theorem shows that this cannot work.

\begin{theorem} 
    There exists a Stochastic Markov Game, $\MDP$, where neither tabular Q-learning nor policy gradient can learn the optimal policy for the leader when
    the follower agent immediately best-responds.
\end{theorem}
\begin{proof}
We consider a slight variation of the iterated prisoner's dilemma discussed in the main text. Consider payoff matrices $L = \begin{pmatrix}0 & -2\\ -1 & -3 \end{pmatrix}$ and $F = \begin{pmatrix} -1 & 0\\ -3 & -2 \end{pmatrix} $ denoting the payoff to the leader and follower agent respectively. The leader chooses the row, the first row denoting ``cooperate'' or $C$ and the second row ``defect'' or $D$, and similarly the follower chooses the column. Notice that these are the standard prisoner's dilemma payoff matrices, except the top and bottom row for the leader have been switched.

Let each agent's observation space be a one-step memory of the \textit{other} agent's previous action, that is, there are three possible observations $o_0$, $o_C$ and $o_D$. At the first step of each episode, both agents observe $o_0$. If at step $t$ the leader cooperates and the follower defects, then at step $t+1$, the leader observes $o_D$ (``other agent defected'') and the follower observes $o_C$ (``other agent cooperated''). We also write $s_{CD}$ for this state if we want to refer to both agents. In particular $s_{CD}$ corresponds to leader reward -2 for the leader and 0 for the follower (top right corner of both matrices). This is a simplification of the setting presented in the main text, but without loss of generality in the case of iterated prisoner's dilemma, and also defines a valid stochastic Markov game in its own right. Let us define an episode of our SMG to be $h$ iterations of the matrix game, where $h$ denotes the horizon or episode length of the game.
As a preliminary, notice that for any leader policy, the follower best-response is always deterministic. This is easy to check.

It is also easy to see that the optimal leader policy is to cooperate on the first step, and to then play tit-for-tat. That is, if the follower cooperated, the leader cooperates in the following step. If the follower defects, the leader defects in return. If the leader plays this policy, then the follower will in turn always cooperate, leading to leader episode reward 0, clearly the optimum. Using the construction in the query-oracle special case of Theorem~\ref{thm:main}, this optimal leader policy can be learned using standard RL algorithms. What we will now show is that if the follower best-responds without that construction, i.e. immediately without queries folded into the leader sample batch, then standard RL algorithms will diverge. This is independent of choice of hyperparameters, but as a matter of principle. 

Intuitively, the problem is one of a missing counterfactual: Notice that for the leader tit-for-tat Stackelberg equilibrium, it is essential that the leader commits to defecting if the follower defects. But notice also that when the leader plays this tit-for-tat policy, the follower will always best respond, and so the leader will never actually see a follower defection. But this also means that it cannot accumulate a gradient for this hypothetical behavior.

We now make this formal. 
Consider first the case of (tabular) Q-learning. Let $q(s,a)$ be the (leader's) Q-value of taking action $a$ in state $s$. We let $\alpha$ denote the learning rate and $\gamma$ the discount factor. Given an experience $(s,a,r,s')$ we update Q-values as follows:
\begin{equation}
    q(s,a) \hspace{0.3cm} \leftarrow  \hspace{0.3cm}(1-\alpha) \cdot q(s,a) \hspace{0.2cm}+ \hspace{0.2cm} \alpha \cdot \big( r + \gamma \max q(s',.) \big)
\end{equation}

As a convenient shorthand and as a slight abuse of notation, we will define $\theta$ as follows:
\begin{align}
    \theta_s =     \begin{cases}
      0 & \text{if $q(s,D) \leq q(s,C)$}\\
      1 & \text{if $q(s,D) > q(s,C)$}
    \end{cases} 
\end{align}
In words, we let $\theta_s = 1$  denote that the current leader policy given the $q(s,a)$ values will defect in state $s$, and 0 if the leader will cooperate in state $s$. We can then write $\theta = (\theta_0, \theta_C, \theta_D)$ for the entire leader policy induced by the current Q-table. $\theta = (0,0,0)$ would denote a leader policy that always cooperates, $\theta = (1,1,1)$ denotes a leader always defecting, and $\theta = (0,0,1)$ denotes the (optimal) tit-for-tat strategy.

Now consider the case of tabular Q-learning with parameter noise exploration. 
In this, we collect experiences from any of the eight possible deterministic leader policies. Note also that the leader action on the initial step does not affect the follower's best response strategy; and it does not influence Q-table updates for the non-initial observations $o_C, o_D$ (because $o_0$ will never be revisited and so the reward generated from $o_0$ can never appear in a Q-table update or indeed in a reward-to-go calculation in a policy gradient algorithm). We can therefore disregard the leader's initial action and for brevity focus only on the four cases $\theta = (\star,0,0)$, $\theta = (\star,0,1)$,  $\theta = (\star,1,0)$ and $\theta = (\star,1,1)$. It is easy to see that for $\theta = (\star,0,1)$ the follower best-response is to always cooperate, and for the other three cases it is to always defect. We may therefore encounter experiences of the following form:
\begin{align*}
    \theta = (\star,0,0)& &\rightarrow& &(o_D, C, -2, o_D) \\
    \theta = (\star,0,1)& &\rightarrow& &(o_C, C, 0, o_C) \\
    \theta = (\star,1,0)& &\rightarrow& &(o_D, C, -2, o_D) \\
    \theta = (\star,1,1)& &\rightarrow& &(o_D, D, -3, o_D)\\
\end{align*}
It is easy to see that under usual Q-learning update rules and for any choice of learning rate, we will have that in the limit $q(o_D,C) = -2 \cdot g$ (lines 1, 3) and $q(o_D,D) = -3 \cdot g$ (line 4) where $g = \frac{1-\gamma^{h/2}}{1-\gamma}$ is a term from the discount factor $\gamma$. Crucially we have that $q(o_D,C) > q(o_D,D)$, and therefore the policy will converge toward $\theta = (\star, \star, 0)$, which is not optimal. This holds for any choice of learning rate, discount factor and exploration parameters (as any mix of the above trajectories will lead to this).

For the $\epsilon$-greedy case, let $\theta^\epsilon_s = \theta_s + (-1)^{\theta_s} (\epsilon/2)$. That is, if our current Q-table induces the deterministic policy $\theta$, then $\theta^\epsilon_s$ gives the probability of choosing action $D$ in state $s$ in the $\epsilon$-greedy case. It is easy to see that for sufficiently small $\epsilon$ and $\theta^\epsilon_s = (\star, \epsilon, 1-\epsilon)$ the follower best-response is still to always cooperate, and for any other $\theta^\epsilon_s $ the follower best-response is to always defect.
Therefore in particular, no matter which way a particular leader action is sampled, the follower will best-respond in the same way (only depending on the leader policy as a whole, not the particular leader action sampled). In turn this means that $q(o_D,C)$ can only continue to accumulate $-2$ terms, and $q(o_D,D)$ can only continue to accumulate $-3$ terms, and the policy will converge toward $\theta = (\star, \star, 0)$, which is not optimal.

To show this for policy gradient, let the leader policy be parametrized by $theta$ as above, i.e. let $\theta_o$ be the probability that the leader policy defects given observation $o$, and $1-\theta_o$ the probability that the leader cooperates given $o$. Recall the basic REINFORCE gradient update rule:
\begin{equation}
    \theta \leftarrow \theta + \alpha G_t \nabla_\theta \ln \pi_\theta(a_t \vert o_t)
\end{equation}
Here $G_t$ denotes the (discounted) ``reward to go'', i.e. $G_t = r_t + \gamma r_{t+1} + \gamma^2 r_{t+2} \ldots$ as usual. A very similar argument as in the Q-learning case now holds to show that the reward-to-go from cooperating when observing $o_D$ will always be larger in expectation than the reward-to-go from defecting, because $r_t$ when defecting is smaller than $r_t$ when cooperating given $o_D$ and the remainder of the sum in $G_t$ is the same in expectation. This in turn pushes gradients toward cooperation, and away from the optimal tit-for-tat policy.
\end{proof}
The above holds for tabular Q-learning and basic policy-gradient with direct parametrization, but likely can be extended to further RL algorithms such as DQN or actor-critic. 

Notice the key difference in the query-oracle POMDP construction: In this, the oracle must query the leader policy for its action given $o_D$ at least once in the initial ``oracle'' segment of the episode. That action therefore sees as its reward to go the reward from the entire final segment, i.e. the entire episode reward of the original Markov game. Intuitively, the leader gets to see at least one experience where it retaliates on a follower defection and this leading to an entire episode of cooperation and good rewards. Without the oracle query, the leader never gets to see this, and cannot learn from it. It may still see experiences where it retaliates for defection, but these will be from within the actual episode, will not influence follower behavior, and will lead to strictly worse rewards than cooperating. Finally, it is also clear that this only applies to typical RL algorithms that learn on taking actions in individual steps. Approaches that learn on the policy space as a whole, such as evolutionary strategies, are not affected by this (as indeed they never look at individual steps and actions at all).

\section{Necessary Conditions for Stackelberg Convergence}
\label{sec:ap:necessary}
It may also be interesting to consider the inverse direction of Theorem~\ref{thm:main}, i.e. what are necessary conditions that follow from Stackelberg convergence. The resulting theorem is not very strong, but still informative, as it suggests avenues for future research. Recall that in Theorem~\ref{thm:main} we map a Markov game to a single-agent RL problem (POMDP) for the leader. In the general case this is simply taking the leader's view of the original Markov game as-is, and in the query-oracle special case we construct a POMDP that incorporates oracle queries. We then show that a solution to the leader's POMDP together with the follower best-response forms a Stackelberg Equilibrium.

Consider now the reverse: Suppose we are given some mapping from Markov game to leader POMDP, and a guarantee that no matter the original Markov game, an optimal solution to the leader POMDP it maps to forms part of a Stackelberg equilibrium. What needs to be true of any such mapping? We formulate this here in a slightly more general manner, in that we also allow an additional (not necessarily identity) mapping between leader policies in the Markov game and the POMDP.

\begin{theorem}[Necessity]
Suppose we are given mappings $\mathcal L: \mathcal M \mapsto \mathcal L(\mathcal M)$ and $l: \Pi_{\mathcal L} \rightarrow \Pi_{\mathcal M}$. $\mathcal L$ maps any Markov Game $\mathcal M$ to a single-agent RL problem, and $l$ maps policies in $\mathcal L$ to policies in $\mathcal M$. Furthermore suppose that whenever a policy $\pi_{L,\mathcal L}$ optimally solves $\mathcal L(\mathcal M)$, then $l(\pi_{L,\mathcal L})$ together with $\mathcal E(l(\pi_{L,\mathcal L}))$ are a Stackelberg equilibrium in $\mathcal M$. Then the following two conditions must be true of $\mathcal L$ and $l$.
\begin{enumerate}
    \item The leader reward in $\mathcal L$ is maximized by the same choice of strategy as the leader reward in $\mathcal M$ when followers play $\mathcal E(\pi_L)$, i.e.
    $$ l\big( \argmax r_{L, \mathcal L}(\pi_L)\big) \subseteq \argmax r_{L, \mathcal M}(\pi_L, \mathcal E(\pi_L))$$
    \item $\mathcal L$ implements a follower equilibrium oracle $\mathcal E(\pi_L)$
\end{enumerate}
\end{theorem}
\begin{proof}[Proof (Sketch)] The first condition immediately follows from the problem statement. For the second condition, consider that given full freedom in choosing $\MDP$, we can construct $\MDP$ so as to let $\mathcal E$ be any arbitrary function from leader to follower policy space. Similarly, we can choose $\MDP$ so that $r_L$ is any arbitrary function. Both of these follow from cardinality arguments, and the observation that since Markov games may be partially observable we are essentially unrestricted in the complexity of the Markov game we choose to construct even for small strategy spaces. Since by the first condition $\mathcal L$ needs to compute $r \circ \mathcal E$, both of which can be arbitrary, it thus also needs to compute $\mathcal E$.
\end{proof}
The main difference to the conditions in Theorem~\ref{thm:main} is that we can only show that the argmax of the leader reward needs to be that of the original Markov game, not that the rewards need to be identical. This is in a way trivial (of course Theorem~\ref{thm:main} still holds if we scaled leader rewards in the leader learning problem by a constant factor), but it also suggests that reward shaping may be a viable technique to accelerate leader learning, potentially still with provable Stackelberg equilibrium guarantees.

\section{Further Experiment Details: Iterated Matrix Games}
\label{sec:ap:experiment}
\paragraph{Environments} We use the 12 canonical symmetric matrix games identified in \cite{robinson2005topology} and also used by \cite{balaguer2022good}. We construct Markov games from these matrices by concatenating multiple iterations into an episode, and giving both agents one-step memory of both agents' action in the previous step. We use $n=10$ steps per episode. Table~\ref{tab:matrices} shows the payoff matrices for all the Markov games, reported on the same scale as the figures. During training, we scale rewards to be centered at 0, i.e. taking values $-1.5$, $-0.5$, $0.5$, $1.5$, but we report results offset to match the reward scales used by \cite{balaguer2022good}. This has no effect on comparability of results.

\paragraph{Algorithm.}
We focus on the contextual policy meta-learning approach described in subsection \ref{sec:contextual} for followers, and standard RL for the leader: At the beginning of each episode, the leader is queried (as part of the episode rollout) for its action in each possible state of the environment. Its responses are then concatenated to the follower observation. In a pre-training phase, we train the follower against randomly sampled leader policies. In the main training phase, we then train the leader against the follower meta-policy. Algorithm~\ref{alg:contextual} in the Appendix details this in pseudo-code.
An advantage of the generality of our framework is that it is agnostic to which specific RL algorithm is used. We generally use a standard policy gradient (PG) algorithm for the followers, although our results do not depend on this specific choice. 

Algorithm~\ref{alg:contextual} details the two-phase learning algorithm we use. In all the experiments shown in the main text, we use policy gradient to train the follower meta-policy in the pre-training loop. We use PG \citep{sutton1999policy}, PPO \citep{schulman2017proximal} and DQN \citep{mnih2013playing,mnih2015human} in the main training loop, as indicated in the respective figures. We use linear models, and disable exploration in the leader policy while pre-training the follower and vice versa. Table~\ref{tab:hyperparam} lists the hyperparameters used for each of these algorithms. Any hyperparameters not listed were left at default values in \textbf{rllib} version 2.0.0. All experiments were run with a single rollout worker (per experiment), and using Torch.

\paragraph{Equilibrium Verification.} At the end of every experiment, we freeze the leader policy and further train the follower policy for $n=50$ iterations. Unlike in the pre-training phase, we here train them only against the specific leader policy trained in the main training loop. This is to further verify that the policies indeed form a Stackelberg equilibrium, and in particular that the follower meta-policy is best-responding to the trained leader. If this is the case, we should not see any change in leader or follower performance in this post-training phase. If the follower meta-policy was \textit{not} already best-responding to the leader, we may see an increase in follower performance during this post-training phase. In all of the experiments in this paper (except the ones designed to show failure modes) we see no follower improvement, i.e. behavior consistent with a Stackelberg equilibrium. This is not shown in the training curves in the figures, but can be reproduced from the source code.

\paragraph{Implementation and Environment.}
All experiments were implemented using Ray / RLlib 2.0.0 \citep{liang2018rllib}. Experiments were run on recent Intel Xeon processors with a single core and 2GB RAM per experiment.

\paragraph{Hyperparameter Tuning.}
Learning rates and batch sizes were tuned using grid search, with some additional tuning using the HyperOpt Python package~\citep{bergstra2013making}, yielding no further improvement however.

\section{Further Experiment Details: Atari 2600}
\label{sec:ap:atari}

\paragraph{Environment.} We modify the Atari 2600 game ``Space Invaders''. We read from emulator RAM to detect when a shot has been fired, and by which player. Separately in a Python wrapper we keep a count of how many shots each player has available. We decrement this whenever we detect that the player fired a shot. If the Python variable keeping track of the available bullets reaches zero, we overwrite the player action that is fed to the Atari emulator to not-firing. Both players start with zero available bullets, but we increment the bullets available to player 1 at stochastic intervals for up to a total of five times per episode. 

We implement a bilateral trade between agents: The selling agent may offer a price, and the buying agent may choose to accept this price. 

\paragraph{Neural Network Architecture.} This is implemented by augmenting both action and observation space, both providing a dictionary of both the underlying Atari action/observation, as well as the new economic action and observations. 

The action space contains  the original Atari action, as well as the trading action. For the seller, the trading action is picking one of several discrete price points, where we choice $n=5$ price points ranging from 0 to 1 in 0.25-step increments. For the buyer, instead of giving a discrete buy / don't-buy action, we let the buyer policy set a maximum price it is willing to buy. If the offered sales price is below the maximum buying price of the buyer, the trade happens, and the price paid is that set by the selling agent. It is easy to see that this is equivalent to letting the buying agent observe the price offer and respond with acceptance or rejection. We chose this implementation as it makes implementing the follower oracle easier when the buyer is Stackelberg leader, but it does not affect the outcome. 

In the observation space, we provide a Dictionary to each agents containing both the original Atari 2600 image observation, as well as all the relevant economic information (number of bullets the agent currently has available, if applicable price offered by the other agent, whether a trade is current being proposed). In the neural network, we run these economic features through a separate fully connected layer, which feeds into a joint logits layer. The Atari input is run through default RLlib CNN and fully connected layers. Figure~\ref{fig:nn_architecture} shows this neural network architecture as a diagram.

\begin{figure}
\begin{center}
\includegraphics[width=0.7\textwidth]{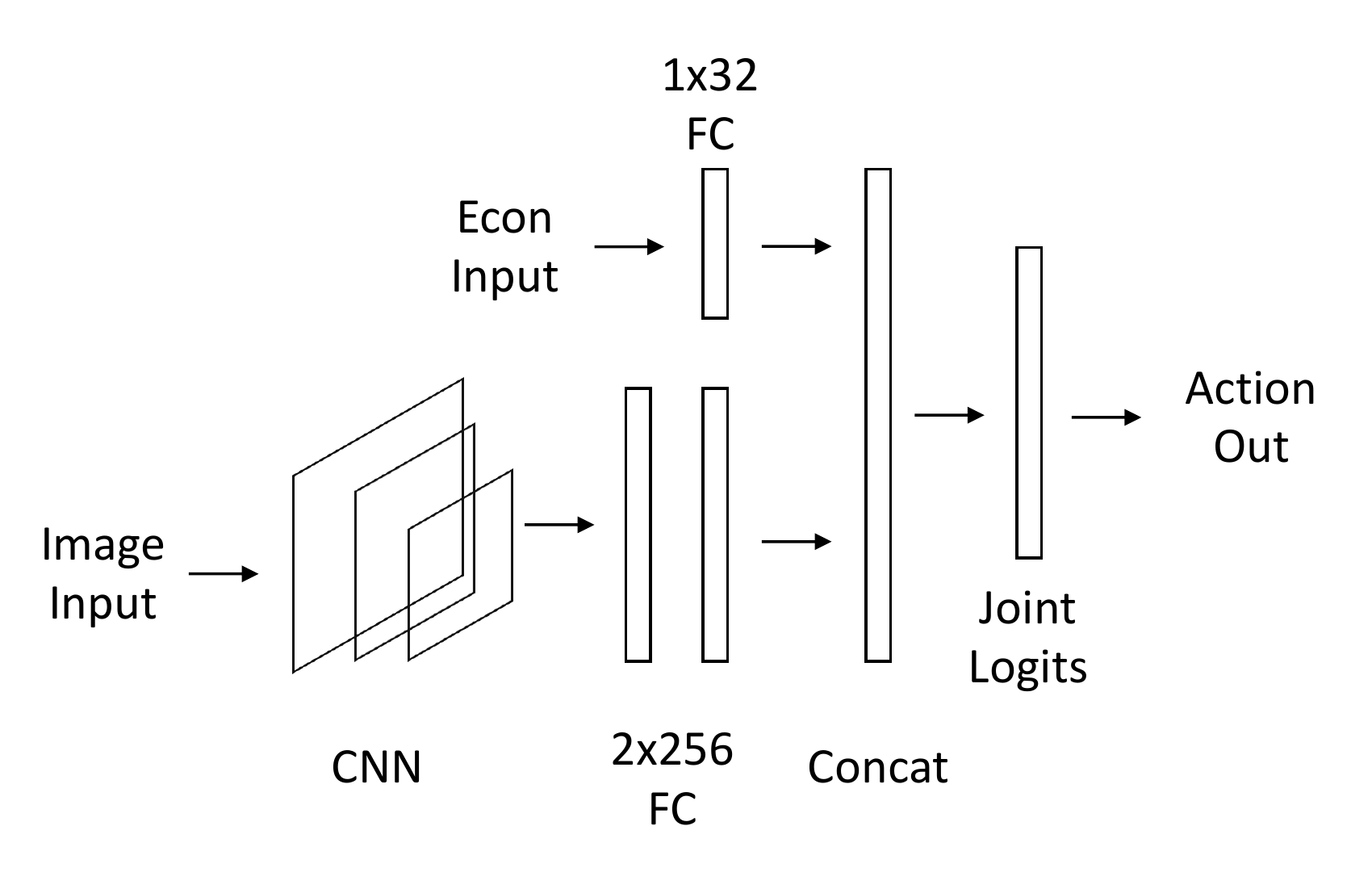}
\end{center}
\caption{Neural Network Architecture used in the Atari 2600 bilateral trade experiments.}
\label{fig:nn_architecture}
\end{figure}

\paragraph{Algorithm.} We use standard PPO for both the leader and follower. Hyperparameters were taken from RLlib tuned examples and are listed at the end of Table~\ref{tab:hyperparam}. For convenience, we initialize weights of the CNN and default RLlib FC layers to weights obtained from training agents in the unmodified game. This speeds up training, but is not strictly necessary. We utilize the same Meta-RL approach as we do in the iterated matrix game experiments: We first train a meta-follower. In this phase, we let the gameplay actions of the leader agent be controlled by an agent trained on the unmodified game, but we randomize the leader's economic actions. Once this meta-follower training has finished, we train the leader. In this phase, the meta-follower weights are frozen, and only the leader policy is trained. In the Atari experiments, we let the meta-follower query the leader immediately before each trade rather than at the start of the episode, as this allows us to fold the queries into the trading exchange.

\section{Further Details on Performance Comparisons}
\label{sec:ap:performance}
In Figure~\ref{fig:allmatrices} we compare our Meta-RL approach with the PPO+Q-learn approach of \cite{brero2022learning} and the ES-MD approach of \cite{balaguer2022good}. 

For \cite{brero2021learning}, we implement follower Q-learning using information therein. Hyperparameters for both the leader and the follower were tuned using the HyperOpt package~\citep{bergstra2013making}. In Figure~\ref{fig:allmatrices} we plot learning curves up to 200k timesteps, as our approach converges before that point. We show in Figure~\ref{fig:allmatrices_2M} learning curves until 2M timesteps. We can see that in some cases PPO+Q-learn eventually converges to the optimum, while in the majority of cases this still has not happened by 2M timesteps. 

For \cite{balaguer2022good}, we estimate their performance from Figure~2 therein. Notice that that figure is \textit{not} a learning curve, but represents a single inner loop at the end of their training procedure. In the ES-MD case, \cite{balaguer2022good} report their performance after 1.28 billion environment steps. In the Diff-MD case, a comparison of sample complexity is difficult, as that approach uses a description of the environment rather than sample access. The closest we can come to a like-for-like comparison is noting that \cite{balaguer2022good} report performance for Diff-MD after 500k computed expected episode returns with 10-step episodes. In some sense this could be seen to be equivalent to 5M environment steps as a lower bound.

\begin{figure}
\begin{center}
\includegraphics[width=\textwidth]{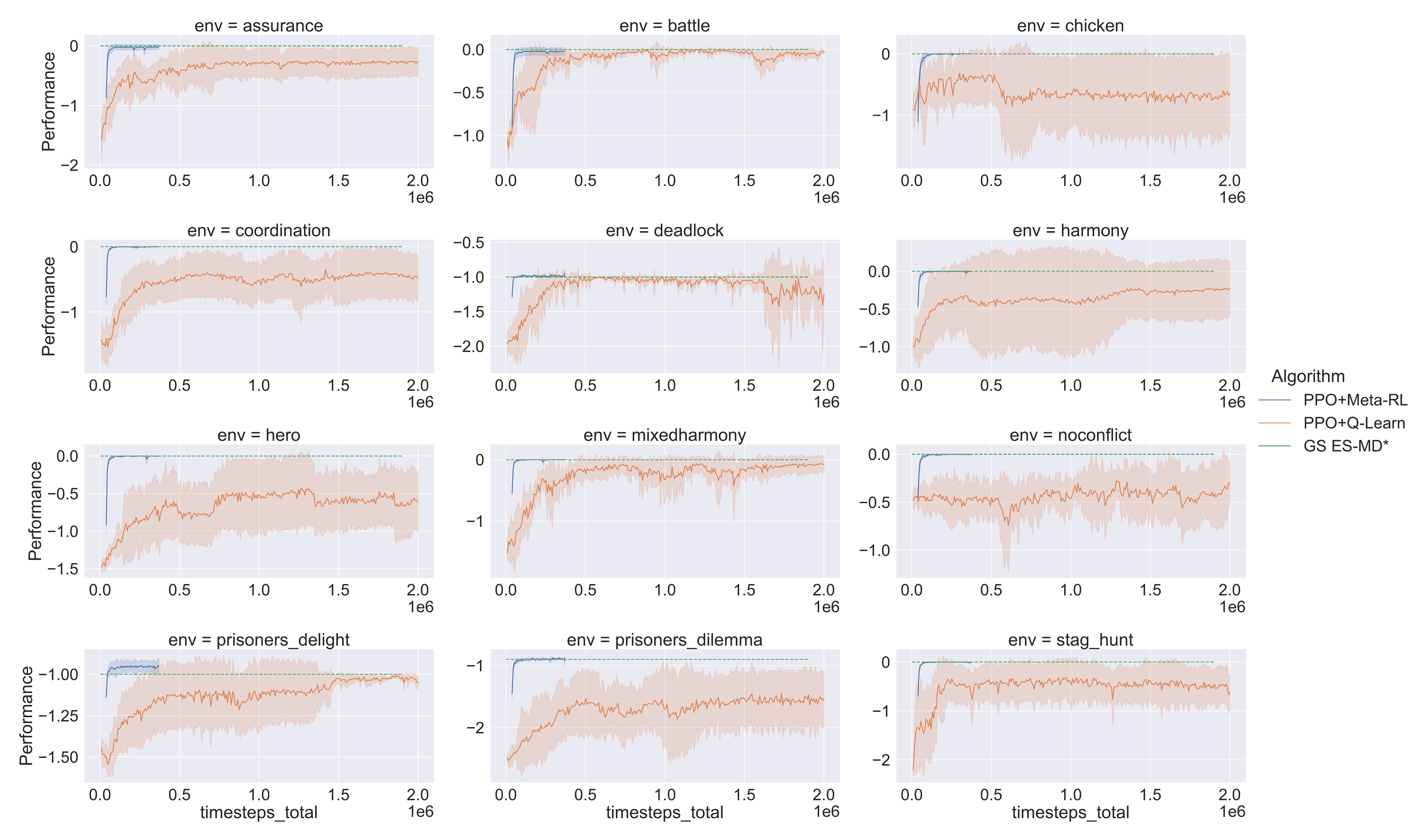}
\end{center}
\caption{Performance on symmetric matrix games (see Figure~\ref{fig:allmatrices}) up to 2M timesteps.}
\label{fig:allmatrices_2M}
\end{figure}

\begin{algorithm}
\caption{Contextual Policy}\label{alg:contextual}
\begin{algorithmic}
\STATE \textbf{Pre-Training Loop}
\STATE Initialize follower policy $\pi_F$
\FOR{each pre-training iteration}
    \FOR{each episode per sample batch}
        \STATE Sample a random leader policy $\pi_L^r$
        \FOR{each $o_L \in O_L$}
            \STATE Query $\pi_L^r$ for $\pi_L^r(o_L)$    
        \ENDFOR
        \STATE Set context $\omega = \pi_L^r(o_L), o_L \in O_L$  
        \FOR{each episode step}
            \STATE Return $o_{L,t}$ to leader, $(\omega, o_{F,t})$ to follower 
            \STATE Step environment using $a_{L,t} = \pi_L^r(o_L),  a_{F,t} = \pi_F(\omega, o_{F,t})$
        \ENDFOR
    \ENDFOR
    \STATE Update follower policy $\pi_F$ using collected sample batch using PG/PPO/DQN
\ENDFOR
\STATE \textbf{Main Training Loop}
\STATE Initialize leader policy $\pi_L$
\FOR{each training iteration}
    \FOR{each episode per sample batch}
        \FOR{each $o_L \in O_L$}
            \STATE Query $\pi_L$ for $\pi_L(o_L)$    
        \ENDFOR
        \STATE Set context $\omega = \pi_L(o_L), o_L \in O_L$  
        \FOR{each episode step}
            \STATE Return $o_{L,t}$ to leader, $(\omega, o_{F,t})$ to follower 
            \STATE Step environment using $a_{L,t} = \pi_L^r(o_L),  a_{F,t} = \pi_F(\omega, o_{F,t})$
        \ENDFOR
    \ENDFOR
    \STATE Update leader policy $\pi_L$ using collected sample batch using PG/PPO/DQN
\ENDFOR
\end{algorithmic}
\end{algorithm}

\begin{table*}[ht] \label{tab:matrices}
  \centering
  \caption{Payoff Matrices used in the matrix-game experiments}

    \begin{tabular}{ccc}
    \toprule
    \multicolumn{3}{l}{\textbf{Iterated Matrix Games (Figure~\ref{fig:allmatrices} etc.)}} \\
    \midrule
    \textbf{Name} & \textbf{Leader Payoff} & \textbf{Follower Payoff} \\
    \midrule
    prisoners dilemma & $\begin{pmatrix} -1 & -3 \\ 0 & -2 \end{pmatrix}$ & $\begin{pmatrix} -1 & 0 \\ -3 & -2\end{pmatrix}$ \\
    \midrule
    stag hunt & $\begin{pmatrix} 0 & -3 \\ -1 & -2 \end{pmatrix}$ & $\begin{pmatrix} 0 & -1 \\ -3 & -2\end{pmatrix}$ \\
    \midrule
    assurance & $\begin{pmatrix} 0 & -3 \\ -2 & -1 \end{pmatrix}$ & $\begin{pmatrix} 0 & -2 \\ -3 & -1\end{pmatrix}$ \\
    \midrule
    coordination & $\begin{pmatrix} 0 & -2 \\ -3 & -1 \end{pmatrix}$ & $\begin{pmatrix} 0 & -3 \\ -2 & -1\end{pmatrix}$ \\
    \midrule
    mixedharmony & $\begin{pmatrix} 0 & -1 \\ -3 & -2 \end{pmatrix}$ & $\begin{pmatrix} 0 & -3 \\ -1 & -2\end{pmatrix}$ \\
    \midrule
    harmony & $\begin{pmatrix} 0 & -1 \\ -2 & -3 \end{pmatrix}$ & $\begin{pmatrix} 0 & -2 \\ -1 & -3\end{pmatrix}$ \\
    \midrule
    noconflict & $\begin{pmatrix} 0 & -2 \\ -1 & -3 \end{pmatrix}$ & $\begin{pmatrix} 0 & -1 \\ -2 & -3\end{pmatrix}$ \\
    \midrule
    deadlock & $\begin{pmatrix} -2 & -3 \\ 0 & -1 \end{pmatrix}$ & $\begin{pmatrix} -2 & 0 \\ -3 & -1\end{pmatrix}$ \\
    \midrule
    prisoners delight & $\begin{pmatrix} -3 & -2 \\ 0 & -1 \end{pmatrix}$ & $\begin{pmatrix} -3 & 0 \\ -2 & -1\end{pmatrix}$ \\
    \midrule
    hero & $\begin{pmatrix} -3 & -1 \\ 0 & -2 \end{pmatrix}$ & $\begin{pmatrix} -3 & 0 \\ -1 & -2\end{pmatrix}$ \\
    \midrule
    battle & $\begin{pmatrix} -2 & -1 \\ 0 & -3 \end{pmatrix}$ & $\begin{pmatrix} -2 & 0 \\ -1 & -3\end{pmatrix}$ \\
    \midrule
    chicken & $\begin{pmatrix} -1 & -2 \\ 0 & -3 \end{pmatrix}$ & $\begin{pmatrix} -1 & 0 \\ -2 & -3\end{pmatrix}$ \\
    \midrule
    \multicolumn{3}{l}{\textbf{Single-Shot Matrix Game (Appendix~\ref{sec:ap:limitations})}} \\
    \midrule
    battle of the sexes & $\begin{pmatrix} 2 & 0 \\ 0 & 1 \end{pmatrix}$ & $\begin{pmatrix} 1 & 0 \\ 0 & 2\end{pmatrix}$ \\
    \midrule
    \multicolumn{3}{l}{\textbf{Modified Prisoner's Dilemma (Theorem~\ref{thm:divergence})}} \\
    \midrule
    prisoners dilemma modified & $\begin{pmatrix} 0 & -2 \\ -1 & -3 \end{pmatrix}$ & $\begin{pmatrix} -1 & 0 \\ -3 & -2\end{pmatrix}$ \\
    \midrule

    \end{tabular}

\end{table*}

\begin{table*}[ht] \label{tab:hyperparam}
  \centering
  \caption{Hyper-Parameter Configuration Table}
    \begin{tabular}{lr|lr}
    \toprule
    \multicolumn{4}{l}{\textbf{Follower Policy Gradient}} \\
    \midrule
    \textbf{Hyper-Parameter} & \textbf{Value} & \textbf{Hyper-Parameter} & \textbf{Value} \\
    algorithm & PG & rollout\_fragment\_length & 100 \\
    lr & 0.02 & train\_batch\_size & 100 \\
    iterations & 500 & batch\_mode & complete\_episodes \\
    \midrule
    \multicolumn{4}{l}{\textbf{Leader Policy Gradient}} \\
    \midrule
    \textbf{Hyper-Parameter} & \textbf{Value} & \textbf{Hyper-Parameter} & \textbf{Value} \\
    algorithm & PG & rollout\_fragment\_length & 100 \\
    lr & 0.156 & train\_batch\_size & 100 \\
    iterations & 1200 & batch\_mode & complete\_episodes \\
    \midrule
    \multicolumn{4}{l}{\textbf{Leader PPO}} \\
    \midrule
    \textbf{Hyper-Parameter} & \textbf{Value} & \textbf{Hyper-Parameter} & \textbf{Value} \\
    algorithm & PPO & rollout\_fragment\_length & 1000 \\
    lr & 0.008 & train\_batch\_size & 1000 \\
    entropy\_coeff & 0.0 & sgd\_minibatch\_size & 1000 \\
    iterations & 500 & batch\_mode & complete\_episodes \\
    \midrule
    \multicolumn{4}{l}{\textbf{Leader DQN}} \\
    \midrule
    \textbf{Hyper-Parameter} & \textbf{Value} & \textbf{Hyper-Parameter} & \textbf{Value} \\
    algorithm & SimpleQ & rollout\_fragment\_length & 10 \\
    lr & 0.001 & train\_batch\_size & 1024 \\
    learning\_starts & 5000 & exploration\_type & ParameterNoise \\
    exploration\_initial\_stddev & 1.0 & exploration\_random\_timesteps & 0 \\
    iterations & 20000 & batch\_mode & complete\_episodes \\
    \midrule
    \multicolumn{4}{l}{\textbf{Atari PPO}} \\
    \midrule
    \textbf{Hyper-Parameter} & \textbf{Value} & \textbf{Hyper-Parameter} & \textbf{Value} \\
    train\_batch\_size & 5000 & rollout\_fragment\_length & 100 \\
    sgd\_minibatch\_size & 100 & num\_sgd\_iter & 10 \\
    lambda & 0.95 & kl\_coeff & 0.5 \\
    clip\_param & 0.1 & vf\_clip\_param & 10.0 \\
    entropy\_coeff & 0.01 & lr & 0.001 \\
    num\_rollout\_workers & 10 & num\_envs\_per\_worker & 5 \\
    \bottomrule
    \end{tabular}
\end{table*}

\end{document}